\documentclass[%
 reprint,
superscriptaddress,
 amsmath,amssymb,
 aps,
]{revtex4-2}

\usepackage{amsmath}
\usepackage{amssymb}
\usepackage{graphicx}
\usepackage{dcolumn}
\usepackage{bm}
\usepackage{hyperref}
\usepackage{tikz}
\usepackage{quantikz}
\usepackage{svg}
\usepackage{algorithm2e}
\usepackage{bbold}
\usepackage{relsize}
\usepackage{amsthm}
\usepackage{comment}
\usepackage{ulem}

\newtheorem{theorem}{Theorem}
\newtheorem{defn}{Definition}
\newtheorem{corrollary}[]{Corollary}
\newtheorem{lemma}{Lemma}
\makeatletter
\newcommand*{\rom}[1]{\expandafter\@slowromancap\romannumeral #1@}
\makeatother

\newcommand{\erf}{\text{erf}}

\newcommand{\poly}{\text{poly}}
\newcommand{\polylog}{\text{polylog}}

\newcommand{\ceil}[1]{\left\lceil #1 \right\rceil}

\renewcommand\Re{\operatorname{Re}}
\begin{document}

\preprint{APS/123-QED}

\title{An efficient quantum algorithm for generation of {\it ab initio} $n$-th order susceptibilities for non-linear spectroscopies}

\author{Tyler Kharazi}
\email{kharazitd@berkeley.edu}
\altaffiliation{Authors contributed equally}
\affiliation{Department of Chemistry, University of California, Berkeley}%

\author{Torin F. Stetina}
\email{torin.stetina@gmail.com}
\altaffiliation{Authors contributed equally; Currently at IonQ, Inc.}
\affiliation{Simons Institute for the Theory of Computing, Berkeley, California, 94704, USA}
\affiliation{Berkeley Quantum Information and Computation Center, University of California, Berkeley, California, 94720, USA}

\author{Liwen Ko}
\affiliation{Department of Chemistry, University of California, Berkeley}%

\author{Guang Hao Low}
\affiliation{Azure Quantum, Microsoft, Redmond, Washington 98052, USA}

\author{K. Birgitta Whaley}
\affiliation{Department of Chemistry, University of California, Berkeley}
\affiliation{Berkeley Quantum Information and Computation Center, University of California, Berkeley, California, 94720, USA}

\date{\today}

\begin{abstract}

We develop and analyze a fault-tolerant quantum algorithm for computing $n$-th order response properties necessary for analysis of non-linear spectroscopies of molecular and condensed phase systems. We use a semi-classical description in which the electronic degrees of freedom are treated quantum mechanically and the light is treated as a classical field. The algorithm we present can be viewed as an implementation of standard perturbation theory techniques, focused on {\it ab initio} calculation of $n$-th order response functions. We provide cost estimates in terms of the number of queries to the block encoding of the unperturbed Hamiltonian, as well as the block encodings of the perturbing dipole operators. Using the technique of eigenstate filtering, we provide an algorithm to extract excitation energies to resolution $\gamma$, and the corresponding linear response amplitude to accuracy $\epsilon$ using ${O}\left(N^{6}\eta^2{{\gamma^{-1}}\epsilon^{-1}}\log(1/\epsilon)\right)$ queries to the block encoding of the unperturbed Hamiltonian $H_0$, in double factorized representation. Thus, our approach saturates the Heisenberg $O(\gamma^{-1})$ limit for energy estimation and allows for the approximation of relevant transition dipole moments. These quantities, combined with sum-over-states formulation of polarizabilities, can be used to compute the $n$-th order susceptibilities and response functions for non-linear spectroscopies under limited assumptions using $\widetilde{O}\left({N^{5n+1}\eta^{n+1}}/{\gamma^n\epsilon}\right)$ queries to the block encoding of $H_0$. 
\end{abstract}



\maketitle
\section{\label{sec:level1}Introduction}
A large focus of quantum algorithms has been on the determination of ground state properties of a given Hamiltonian, with a much smaller effort on extracting excited state properties \cite{caiQuantumComputationMolecular2020,huangVariationalQuantumComputation2022,kosugiLinearresponseFunctionsMolecules2020, roggeroLinearResponseQuantum2019, liPerturbationTheoryMethods2023, hait2018accurate}. While accurate estimation of ground state properties of quantum many-body systems is already difficult, excited state properties are often more challenging to calculate, because interior eigenvectors and eigenvalues must be resolved without a direct variational analog, since they are not the extremal eigenvalues of the Hamiltonian describing the system. Classical approaches to solve this problem are often expressed using the generalized variational principle, that target specific excited states~
\cite{FreqDepPol,hait2020excited,shea2020generalized}. While there exist efficient classical algorithms for approximating the vibronic spectra of molecules \cite{ohQuantuminspiredClassicalAlgorithms2024}, classical algorithms for electronic spectroscopy in general suffer from exponential scaling. 

On the other hand, understanding non-ground state properties such as transition dipole moments and energy gaps allows better predictions of material properties, and is essential for predicting and analyzing spectroscopic features deriving from light-matter interactions. In this work, we present a quantum algorithm that can be used to accurately compute excited state properties deriving from the interaction of a many-electron system with photons in the semi-classical description, given access to efficient ground state preparation. In the weak field regime and in the electric dipole approximation, we obtain {\it ab initio} estimates of the transition dipole moments corresponding to electronic state transition probabilities, as well as their associated excitation energies.  This allows efficient construction of $n$-th order susceptibilities and hence efficient {\it ab initio} calculation of non-linear spectroscopic responses.  We focus here specifically on the case of molecular spectroscopy, but our results can be applied more generally to any many-body fermionic system and its non-linear response to interaction with a classical light field.

The main contribution of this work is the development and analysis of an algorithm for computing {\it ab initio} transition dipole moments and excitation energies. This algorithm allows the determination of frequency-dependent properties of the response function by using a simple excite-and-filter approach where we first apply the dipole operator and then filter the response to eigenenergies within a particular frequency bin. By using quantum signal processing \cite{lowOptimalHamiltonianSimulation2017} in conjunction with a modified binary search procedure, our algorithm achieves Heisenberg limited scaling~\cite{atia2017fast} for determination of excitation energies. We therefore expect it to be near optimal for the computation of frequency-dependent response properties in the semi-classical regime. Our proposed algorithm has the additional benefit that the core subroutine can be iteratively repeated to compute response functions of any order. As far as we are aware, this is the first work to provide concrete resource estimations for the general problem of estimating linear and nonlinear response functions of arbitrary order.

We focus on estimating transition strengths at specific excitation energies by implementing the transition dipole operator directly on a prepared ground state and filtering the response on a particular energy window. By repeatedly measuring a constant number of additional ancilla qubits, we can obtain systematically improvable estimates to transition dipole moments within this window. After obtaining the relevant transition dipole moments and energies, one can then post-process this data on a classical computer to extract interesting quantities such as oscillator strengths \cite{hilbornEinsteinCoefficientsCross2002} and compute the desired response functions. 

The structure of this paper is as follows. In section \ref{sec:level2}, we introduce the molecular electronic structure Hamiltonian in second quantization, followed by an overview of relevant aspects of electronic spectroscopy within the semi-classical treatment of electron-light interactions and under the electronic dipole approximation. We then briefly introduce the primary quantum subroutines that we use in the main algorithm of this work. In section \ref{sec:level3} we present the algorithm in detail and show the expected asymptotic cost of the computational resources required. We then show how the algorithm can be iteratively applied to find any order of response. In \ref{sec:level4} we compare our results with other proposed quantum algorithms for the computation of polarizability of linear response functions, before concluding and giving an outlook for future work. 

\section{\label{sec:level2}Background }
In this section we first introduce the molecular Hamiltonian (II.A) and electronic spectroscopy in the electric dipole setting (II.B) and then summarize the main algorithmic primitives used in this work. These are the linear combination of unitaries (II.C) and quantum signal processing (II.D).
Unless otherwise noted, we will use $||\cdot||$ to refer to the vector 2-norm when the argument is a vector and the spectral norm, i.e. the largest eigenvalue, when the argument is a matrix. We use the notation $| \cdot |$ to refer to the absolute value if the argument is a scalar and is the set cardinality if the argument is a set. We use the over bar notation $\widetilde{\cdot}$ to refer to approximate quantities. Furthermore, we will assume fault-tolerance throughout this work, and any errors result from approximation techniques or from statistical sampling noise. We analyze the effect of both of these kinds of errors in the text.

\subsection{Molecular Hamiltonian}
We begin by assuming that the system of interest is well approximated by the electronic Hamiltonian under Born-Oppenheimer approximation, with clamped nuclei represented as classical point charges. Starting from a Ritz-Galerkin discretization scheme, where basis functions are spin-orbitals in real space, we obtain the molecular electronic Hamiltonian in second quantization as
\begin{equation}
    H_0 = \sum_{p,q}^{N} T_{pq}\, a_p^\dagger a_q + \sum_{p,q,r,s}^N V_{pqrs}\, a_p^\dagger a_q^\dagger a_r a_s.
    \label{eq:elec_ham2q}
\end{equation}
where $a^{\dagger}$ and $a$ are the fermionic creation and annihilation operators respectively, and $N$ is the number of spin-orbital basis functions. $p,q,r,s$ denote spin-orbital indices unless otherwise noted. {We will work in the atomic units, where the reduced Planck constant $\hbar$, electron mass $m_e$, elementary charge $e$, and permittivity $4\pi\epsilon_0$, are set to $1$.}

The first term in the electronic Hamiltonian encapsulates the one-body interactions, including the kinetic energy of the electrons and the electron-nuclear interaction. The matrix elements are computed by finding the representation of the kinetic energy and electronic-nuclear repulsion terms in the chosen spin-orbital basis, which for the $p$th basis function is denoted as $\phi_p$. 
\begin{equation}
 T_{pq} = \int d^3\textbf{r} \,\, \phi_p^*(\textbf{r})\left(-\frac{1}{2}\nabla^2_\textbf{r}+\sum_\alpha \frac{Z_\alpha}{|\textbf{r}- \textbf{R}_\alpha|}\right)\phi_q(\textbf{r})
\label{eq:kinetic2q}
\end{equation}
where $\mathbf{r} \in \mathbb{R}^{3}$ is the position coordinate vector for an electron, $\mathbf{R}_{\alpha} \in \mathbb{R}^{3}$ is the $\alpha$th nuclear position coordinate vector, $Z_\alpha$ is the integer electric charge of the $\alpha$th nuclei, $\phi_p(\textbf{r})$ is the spin-orbital basis function, and $\nabla^2_{\mathbf{r}}$ is the $3$-dimensional Laplacian operator.

The second term, which is only a function of the electronic positions, is the two-body interaction term. The matrix elements are
\begin{equation}
V_{pqrs}=\iint  d^3\textbf{r} \, d^3\textbf{r}' \, \frac{\phi_p^* (\textbf{r}) \phi_q^* (\textbf{r}')\phi_r(\textbf{r}')\phi_s(\textbf{r})}{|\textbf{r}-\textbf{r}'|},
\label{eq:pot2q}
\end{equation}
which is the representation of the Coulombic repulsion in the MO basis.
Due to the $O(N^4)$ terms in the two-body term as opposed to the $O(N^2)$ terms in the one-body term, it is the dominant contribution to the complexity of simulating the electronic Hamiltonian without further approximation. However, recent techniques based on double factorization {or other tensor factorizations} of the 2-electron operator~\cite{pengLowrankFactorizationElectron2017}, can reduce this to between $O(N^2)$  terms where $N$ scales with number of atoms and $O(N^3)$ terms where $N$ scales with basis set size for a fixed number of atoms~\cite{reiherElucidatingReactionMechanisms2017,babbushLowDepthQuantum2018} or even fewer in practice \cite{leeEvenMoreEfficient2021}. In the following, we will assume $O(N^2)$ scaling for number of terms.

\subsection{Electronic Spectroscopy}
Spectroscopy allows physical systems to be probed via their interaction with light. A very common form of spectroscopy is electronic absorption spectroscopy, for which the UV-Vis frequency range allows probing of molecular electronic transitions (see e.g. Ref. \cite{normanPrinciplesPracticesMolecular2018} Section 5 for a review). A typical molecular electronic absorption spectroscopy experiment consists of a sample of molecules in solution, that is irradiated by broadband light in the visible to ultraviolet range. The absorption spectrum is is calculated from the difference in intensity with and without the sample present.

Spectroscopy experiments yield vital information about the excited state energies and absorption probabilities of the system of interest. Quantities such as the oscillator strengths, static and dynamic polarizabilities, and absorption cross-sections are used to understand the electronic structure of molecular systems and to design materials that respond to light at specific frequencies \cite{mukamel1995principles, gobreEfficientModellingLinear2016}.  

Typical molecular electronic absorption spectroscopy is performed with laser fields that can be represented by classical fields, resulting in a semi-classical description in which the electronic degrees of freedom are treated quantum mechanically and interact with a classical {electromagnetic field.}{ The interaction between charged particles and the electromagnetic field is described by the minimal coupling Hamiltonian
\begin{equation}
    H(t) = H_0 + \sum_\alpha \frac{\big(\mathbf{p}_\alpha-q_\alpha \mathbf{A}(\mathbf{r}_\alpha, t)\big)^2}{2m_\alpha},
\label{Eq:minimal_coupling_H}
\end{equation}
where $\alpha$ indexes over all charged particles. $\mathbf{p}_\alpha$, $q_\alpha$, $\mathbf{r}_\alpha$, and $m_\alpha$ are the  momentum, charge, position, and mass of the $\alpha$-th particle, respectively. $\mathbf{A}$ is the classical vector potential of the electromagnetic field. Note that the Hamiltonian is now time-dependent due to the fact that the classical field is time-dependent. Since the length scales of molecules are typically much smaller than the wavelength of visible light that gives rise to electronic transitions, one can invoke the dipole approximation to rewrite the Hamiltonian of Eq. (\ref{Eq:minimal_coupling_H}) into the dipole -- electric field form as
}
{
\begin{equation}
	H(t) = H_0 - \mathbf{d}\cdot\mathbf{E}(t)
\end{equation}
\cite{Loudon_book,Gerry_book,mukamel1995principles}. Here, $\mathbf{d}$ is the dipole operator, which consists of a nuclear part $\mathbf{d}_n$ and an electronic part $\mathbf{d}_e$. They are defined as
\begin{equation}
	\mathbf{d} = \mathbf{d}_n + \mathbf{d}_e = \sum_\alpha Z_\alpha \mathbf{R}_\alpha - \sum_\beta \mathbf{r}_\beta,
\end{equation}
where $\alpha$ and $\beta$ index over the nuclei and electrons, respectively.
Under the Born-Oppenheimer approximation, the dipole matrix element $\langle n|\mathbf{d}|m\rangle$ for the transition between electronic states $|n\rangle$ and $|m\rangle$ has no contribution from the nuclear part $\mathbf{d}_n$ because the electronic wavefunctions are orthogonal. The electronic contribution to the matrix element is usually written as a product of a nuclear wavefunction overlap and an electronic dipole matrix element $\langle m|\mathbf{d}_e|n\rangle$, where $|m\rangle$ and $|n\rangle$ are the electronic wavefunction. The nuclear wavefunction overlaps give rise to the vibronic structure in electronic spectra \cite{atkins2011molecular}. In this work, we shall ignore the quantum nature of nuclei and therefore neglect the vibronic effects in electronic spectroscopy. Since only the electronic part of the dipole moment, $\mathbf{d}_e$ contribute to the dipole matrix elements in electronic transitions, from now on, we will denote $\mathbf{d}_e$ as $\mathbf{d}$ for notational simplicity.
}

We can now obtain the second-quantized representation of {$\mathbf{d}$} by computing matrix elements in the fixed MO basis for the unperturbed Hamiltonian. Thus
{
\begin{equation}
    d_{i,pq} = -\int d\mathbf{r}\,\phi_p^*(\textbf{r}) r_i \, \phi_q(\textbf{r}) ,
    \label{eq: dipole elems}
\end{equation}
}
gives the representation of the matrix element of the {$i$th} component $d_i$ of the vector operator $\mathbf{d}$ in the spin-orbital basis used in defining the one and two electron integrals for $H_0$, Eqs. \eqref{eq:kinetic2q} - \eqref{eq:pot2q}. Here {$r_i$} denotes the {ith component of the position $\mathbf{r}$ in 3-d space}. Once the basis has been selected and the relevant integrals performed, we then obtain the following second quantized representation of {$\mathbf{d}$}, 
{
\begin{equation}
    \mathbf{d} = \sum_{p,q = 1}^N \mathbf{d}_{p,q} \, a_p^\dagger a_q,
    \label{eq:D2ndquant}
\end{equation}
where $\mathbf{d}_{pq}$ is a 3-vector with components $d_{1,pq}$, $d_{2,pq}$, and $d_{3,pq}$.
}

Once the induced dipole moment {$\mathbf{d}$} is calculated, many quantities of interest, such as {the linear or nonlinear electric susceptibilities}, can be derived from the matrix elements formed by the eigenvectors of the unperturbed Hamiltonian {$\mathbf{d}_{mn} = \bra{m} \mathbf{d}\ket{n}$}. Typically, absorption spectroscopy starts in the electronic ground state (assuming room temperature, $\sim300$K), so we set $m = 0$, and $n$ indexes the set of all excited electronic eigenstates. {Electric susceptibility measures the response of the molecular dipole moment to external electric fields. It is defined generally in the response formalism as
\begin{align}
\begin{split}
    \langle d_i(t)\rangle =& \langle d_i(0)\rangle \\
    &+ \sum_{j=1}^3\int^t_0 dt_1\,\chi_{ij}^{(1)}(t-t_1)E_j(t_1)\\
    &+\sum_{j,k=1}^3\int^t_0 dt_2\int^{t_2}_0 dt_1\,\chi_{ijk}^{(2)}(t-t_2, t_2-t_1)E_j(t_2)E_k(t_1) \\
    &+\cdots,
\end{split}
\end{align}
where $\chi^{(n)}$ is the $n$th order susceptibility \cite{mukamel1995principles}, and the sum of the terms involving $\chi^{(n)}$ is the $n$th order dipole moment, denoted as $d^{(n)}_i$. For weak-light matter interaction, as in the case of most electronic spectroscopy experiments, only the low-order effects are important. For example, the linear susceptibility $\chi^{(1)}$ provides information about the absorption from the ground state. Due to symmetry considerations, the second order susceptibility $\chi^{(2)}$ and all other even order susceptibilities are zero if molecules are isotropically distributed, as in the case of molecules in solution. The third order susceptibility $\chi^{(3)}$ is widely used by nonlinear spectroscopists in experiments such as pump-probe or two-dimensional electronic spectroscopy, and it can be used to describe the dynamics in excited states. Higher order effects such as $\chi^{(5)}$ have also been studied recently \cite{dostal2016coherent}. 
\par
The expressions for $\chi^{(n)}$ can be derived from the perturbative expansion of the time evolution of the density matrix describing the electronic state of the molecule, and they are given by
\begin{align}
\begin{split}
    \chi^{(n)}_{ii_n\cdots i_1}(s_n, \cdots, s_1) = i^n &\text{Tr}\Big(d_i\mathcal{G}(s_n)d^\times_{i_n}\cdots \mathcal{G}(s_1)d^\times_{i_1}\rho_0\Big).
\label{Eq:susceptibility_time_def_general}
\end{split}
\end{align}
$\rho_0=|0\rangle\langle 0|$ is the ground state density matrix. The superscript $\times$ in $d^\times_i$ denotes the commutator superoperator, whose action on a general operator $X$ is defined as $d^\times_i X = [d_i, X]$. The Green's function $\mathcal{G}(s)$ is the time-evolution superoperator, defined by the action 
\begin{equation}
    \mathcal{G}(s)X = e^{-\gamma s}e^{-iH_0 s}Xe^{iH_0 s}\theta(s),
\label{Eq:Greens_func_time}
\end{equation}
where the step function
\begin{equation}
    \theta(s)=
    \begin{cases}
    1\quad ,s>0 \\
    0\quad ,s<0
    \end{cases}
\end{equation}
ensures that $\mathcal{G}(s)$ acts non-trivially only when $s>0$. The exponential factor $e^{-\gamma s}$ in the definition of $\mathcal{G}(s)$ is a phenomenological decay factor that accounts for the decay of electronic coherence and excited state populations, typically due to the interaction with nuclear degrees of freedom or the quantized photon degrees of freedom. For simplicity, we assume a constant decay rate $\gamma$ for all elements $|m\rangle\langle n|$ of the density matrix. However, this assumption can be relaxed straightforwardly.
} 
\par
{Sometimes the susceptibilities are expressed in the frequency domain. We denote the frequency domain susceptibilities as $\alpha^{(n)}$, and it is defined by
\begin{align}
\begin{split}
    d^{(n)}_i(t) = \frac{1}{(2\pi)^n} &\sum_{i_n,\cdots,i_1\in\{1,2,3\}}\int d\omega_n\cdots d\omega_1 \\
    &\alpha^{(n)}_{ii_n\cdots i_1}(\omega_n,\cdots,\omega_1) \\
    &E_{i_n}(\omega_n)\cdots E_{i_1}(\omega_1)e^{-i\omega_s t},
\end{split}
\end{align}
where $\omega_s = \omega_1+\cdots+\omega_n$.
We have used the following definition for the Fourier transform 
\begin{subequations}
\begin{equation}
    f(\omega) = \int dt\, f(t)e^{i\omega t}
\end{equation}
\begin{equation}
    f(t) = \frac{1}{2\pi}\int d\omega\, f(\omega)e^{-i\omega t}.
\end{equation}
\end{subequations}
The frequency-domain susceptibilities $\alpha^{(n)}$ are related to the time-domain susceptibilities $\chi^{(n)}$ by
\begin{align}
\begin{split}
    &\alpha^{(n)}_{ii_n\cdots i_1}(\omega_n,\cdots,\omega_1) = \frac{1}{n!}\\
    &\sum_{P\in S_n} \chi^{(n)}_{ii_{P(n)}\cdots i_{P(1)}}(\omega_{P(1)}+\omega_{P(2)}+\cdots+\omega_{P(n)}, \cdots, \\
    &\qquad\qquad\qquad\qquad \omega_{P(1)}+\omega_{P(2)}, \omega_{P(1)}),
\label{Eq:chi_to_alpha}
\end{split}
\end{align}
where $S_n$ is the symmetric group, and $P$ indexes over all $n!$ possible permutations of the $n$ interactions. This sum takes into account all possible time-orderings of the $n$ interactions.
To obtain $\chi^{(n)}(\omega_1+\cdots+\omega_n, \cdots, \omega_1)$, we first note that by applying Eq. (\ref{Eq:Greens_func_time}), the action of the time-domain Green's function $\mathcal{G}(s)$ on the matrix element $|m\rangle\langle n|$ in the energy eigenstate basis is
\begin{equation}
    \mathcal{G}(s)(|m\rangle\langle n |) = e^{(-i\omega_{mn}-\gamma)s}\theta(s)|m\rangle\langle n|
\label{Eq:Greens_func_time_applied_on_mn}
\end{equation}
where $\omega_{mn}=\omega_m-\omega_n$. 

\par
As an example, using Eqs. (\ref{Eq:susceptibility_time_def_general}) and (\ref{Eq:Greens_func_time_applied_on_mn}), the linear susceptibility in the time-domain is expressed in the sum-over-state form as
\begin{equation}
    \chi^{(1)}_{ij}(s) = \theta(s)\sum_n i d_{i,0n} d_{j,n0} e^{(-i\omega_{n0}-\gamma)s} + \text{c.c.},
\end{equation}
where $d_{i,mn}=\langle m|d_i|n\rangle$. Using Eq. (\ref{Eq:chi_to_alpha}), the linear susceptibility in the frequency domain is 
\begin{equation}
    \alpha^{(1)}_{ij}(\omega) = \sum_m \frac{d_{i,0n}d_{j,n0}}{\omega_{n0}-\omega-i\gamma} + (\omega\rightarrow-\omega^*).
\label{eq:offdiagpolar}
\end{equation}
The notation $(\omega\rightarrow-\omega)^*$ means switching all $\omega$ to $-\omega$ in the previous term and then taking the complex conjugate. 
}
The first-order susceptibility is also known as the polarizability. It is a measure of the amount of change in the electronic configuration when subjected to an external electric field. In particular, the imaginary part of $\alpha^{(1)}$ is proportional to the absorption spectrum.
$\gamma$ is usually defined as the half-width-half-maximum of the peaks, also referred to as the spectral resolution (see \cite{normanPrinciplesPracticesMolecular2018} sec. 7.5). Finite values $\gamma \neq 0$ make Eq.~\eqref{eq:offdiagpolar} well defined when the laser frequency is resonant with a molecular transition, i.e., $\omega = \omega_{n0}$.

Higher-order response functions such as the first hyperpolarizability and higher susceptibilities can be expressed in terms of higher-order moments of the induced dipole operator $\boldsymbol{d}$. The calculation of these quantities also requires computing matrix elements for excited state-to-excited state transitions, not just ground-to-excited-state transitions. These quantities provide valuable information characterizing light-matter interactions.
For example, third-order response functions are used in the modeling of four-wave mixing in non-linear phenomena, which have numerous applications in the modeling of optical physics, see e.g. \cite{chenSimulationTimeFrequencyResolved2021, wangGenerationCorrelatedPhotons2001,slusherSqueezedlightGenerationFourwave1987,sharpingFourwaveMixingMicrostructure2001}.

For {arbitrary} $n$th order spectroscopies, we need in general to compute $O(2^n)$ correlation functions in the time domain (see Ref. \cite{tokmakoff_notes} Eq. 30) {because each of the $n$ interactions involves a commutator, which has two terms (see Eq. (\ref{Eq:susceptibility_time_def_general}))}. Each of these requires an additional $O(n!)$ permutations of the orderings in which we apply the operators (see {Eq. (\ref{Eq:chi_to_alpha})} and Ref \cite{tokmakoff_notes} Eq. 59). For example, the computation of third order response functions in the frequency domain requires the evaluation of 48 correlation functions. {In general, the asymptotic scaling for $n$th order spectroscopies can almost certainly be assumed to be $O(1)$ while also however being significantly larger than a single term. Thus, while
higher-order spectroscopic are of interest~(e.g., fifth-order spectroscopy can reveal details of exciton-exciton coupling~\cite{dostal2016coherent}), many of these higher-order quantities are nevertheless inaccessible with any known classical algorithm with provable performance bounds.}

However, this number of terms may be greatly reduced in practice due to problem symmetries, and often practioners are interested in specific {contributions, not in the sum of all 48 terms because experiments can directly probe specific terms by making use of pulse-ordering and phase matching~\cite{mukamel1995principles}. 
For example, 
by choosing a fixed time ordering for the three pulses, there is no need to permute the three frequencies, so the number of terms reduces to 48/6 = 8. 
Matching the optical phase of input and output pulses and ensuring that the output signals from different terms are projected into different spatial directions allow the number of terms to be reduced to just two terms which are complex conjugates.} 

As {another} example, we consider the commonly encountered third-order response function. 
{Using Eq. (\ref{Eq:susceptibility_time_def_general}), the third order susceptibility in the time domain can be written as
\begin{align}
\begin{split}
    \chi_{ii_3 i_2 i_1}^{(3)}(s_3,s_2,s_1) = \theta(s_3)\theta(s_2)\theta(s_1)\\\sum_{\nu=1}^4 i^3 R_{\nu,ii_3 i_2 i_1}(s_3,s_2,s_1) + \text{c.c.},
\label{Eq:chi3_4_Rs}
\end{split}
\end{align}
where
\begin{subequations}
\begin{equation}
    R_{1,ii_3 i_2 i_1}(s_3,s_2,s_1) = \text{Tr}(d_{i}\mathcal{G}(s_3)d_{i_3}^r\mathcal{G}(s_2)d_{i_2}^r\mathcal{G}(s_1)d_{i_1}^l\rho_0)
\end{equation}
\begin{equation}
    R_{2,ii_3 i_2 i_1}(s_3,s_2,s_1) = \text{Tr}(d_{i}\mathcal{G}(s_3)d_{i_3}^r\mathcal{G}(s_2)d_{i_2}^l\mathcal{G}(s_1)d_{i_1}^r\rho_0)
\end{equation}
\begin{equation}
    R_{3,ii_3 i_2 i_1}(s_3,s_2,s_1) = \text{Tr}(d_{i}\mathcal{G}(s_3)d_{i_3}^l\mathcal{G}(s_2)d_{i_2}^r\mathcal{G}(s_1)d_{i_1}^r\rho_0)
\end{equation}
\begin{equation}
    R_{4,ii_3 i_2 i_1}(s_3,s_2,s_1) = \text{Tr}(d_{i}\mathcal{G}(s_3)d_{i_3}^l\mathcal{G}(s_2)d_{i_2}^l\mathcal{G}(s_1)d_{i_1}^l\rho_0).
\end{equation}
\end{subequations}
The superscripts $l$ and $r$ denote the left- and right-multiplying superoperator, i.e., $A^l X = AX$ and $A^r X = XA$. Each $R_\nu$ can be expressed explicitly in the sum-over-state form. For example, 
\begin{align}
\begin{split}
    R_{1,ii_3 i_2 i_1}&(s_3,s_2,s_1) = \sum_{n,m,l} d_{i_2,0l}d_{i_3,lm}d_{i,mn}d_{i_1,n0}\\
    &e^{(-i\omega_{n0}-\gamma)s_1}e^{(-i\omega_{nl}-\gamma)s_2}e^{(-i\omega_{nm}-\gamma)s_3}.
\end{split}
\end{align}
The four $R_\nu$ can be represented diagrammatically in terms of the double sided Feynman diagrams \cite{mukamel1995principles, tokmakoff_notes}, and they represent different nonlinear physical processes. For example, both $R_1$ and $R_2$ can describe the processes of excited state absorption and stimulated emission, depending on the intermediate states involved. $R_3$ and $R_4$ can describe the processes of ground state bleach. The processes can be further classified into rephasing and non-rephasing pathways, which can be distinguished experimentally using phase matching. Comparison between signals from the rephasing and non-rephasing pathways provides important information regarding the inhomogeneity of the sample. It is common to perform Fourier transforms on the first and the third time variables in $R_\nu(s_3,s_2,s_1)$, and write it as $R_\nu(\omega_3, s_2, \omega_1)$. These quantities are directly related to the spectra obtained in pump-probe spectroscopy or two-dimension electronic spectroscopy experiments, and they reveal the dynamics of electronically excited states.
}
\par
{Using Eqs. (\ref{Eq:chi_to_alpha}) and (\ref{Eq:chi3_4_Rs}), the frequency-domain third order susceptibility $\alpha^{(3)}$ can be written as
\begin{align}
\begin{split}
    &\alpha^{(3)}_{i i_3 i_2 i_1}(\omega_3, \omega_2, \omega_1) = \frac{1}{6} \sum_{P\in S_3} \sum_{\nu=1}^4 \\
    &i^3 R_{\nu,ii_{P(3)} i_{P(2)} i_{P(1)}}(\omega_{P(1)}+\omega_{P(2)}+\omega_{P(3)}, \\
    &\qquad\qquad\qquad\qquad \omega_{P(1)}+\omega_{P(2)}, \omega_{P(1)})\\
    &+(\omega_1,\omega_2,\omega_3\rightarrow -\omega_1, -\omega_2, -\omega_3)^*.
\end{split}
\end{align}
$R_\nu(\omega_3,\omega_2,\omega_1)$ can be expressed explicitly in the sum-over-state form. For example,
\begin{align}
\begin{split}
    i^3 R_{1, i i_3 i_2 i_1}(\omega_1+\omega_2+\omega_3,\omega_1+\omega_2,\omega_1) = \\
    \sum_{n,m,l}d_{i_2,0l}d_{i_3,lm}d_{i,mn}d_{i_1,n0} \frac{1}{(\omega_{n0}-\omega_1-i\gamma)}\\
    \frac{1}{(\omega_{nl}-\omega_1-\omega_2-i\gamma)}\frac{1}{(\omega_{nm}-\omega_1-\omega_2-\omega_3-i\gamma)}.
\label{eq:thirdord}
\end{split}
\end{align}
}




\subsection{Linear Combination of Unitaries}
Block encoding  is a technique to probabilistically implement non-unitary operations on a quantum computer and one approach to constructing a block encoding is the method of linear combination of unitaries (LCU)~\cite{childsHamiltonianSimulationUsing2012}. LCU begins by assuming you have a matrix operation $A$ that has a decomposition into $k$ unitary matrices $U_k$
\begin{equation}
    A = \sum_{l=0}^{k-1}c_l U_l.
\end{equation}
The goal is to embed the matrix operation given by $A$ into a subspace of a larger matrix $U_A$ which acts as a unitary on the whole space. Because of the restriction to unitary dynamics, we require that $||A|| \leq 1$. Since this is difficult to satisfy for a general matrix, a subnormalization factor $\alpha$ is used so that $||A|| \leq \alpha$. The block encoding $U_A$ can be expressed in matrix form
\begin{equation}
    U_A = \begin{pmatrix}
        A/\alpha & *\\
        * & *
    \end{pmatrix}.
\end{equation}

With LCU, it is straightforward to obtain an upper bound on $||A||$ as
\begin{equation}
    \alpha = \sum_{l=0}^{k-1}|c_l|.
\end{equation}
Then, scaling the LCU by $1/\alpha$
\begin{equation}
    A/\alpha = \sum_{l=0}^{k-1}\frac{c_l}{\alpha}U_l \equiv \sum_{l=0}^{k-1}\beta_l U_l,
\end{equation}
we have a matrix that can be block encoded using LCU.

The unitary description of LCU is as follows. We use an ancilla register of $q = O(\log(k))$ qubits and define the oracle which ``prepares" the LCU coefficients
\begin{equation}
    \textsc{prep}\ket{0}^{\otimes q} \rightarrow \sum_{l}\sqrt{\beta_l}\ket{l}
\end{equation}
which is a normalized quantum state and corresponds to the first column of a unitary matrix. Next, we implement the controlled-application of the desired unitaries using the ``select" routine
\begin{equation}
    \textsc{sel}:= \sum_{l}\ket{l}\bra{l}\otimes U_l.
\end{equation}

Then, the block encoding can be expressed as the $n+q$ qubit unitary
\begin{equation}
    U_A := \left(\textsc{prep}^\dagger\otimes I_n\right)\cdot \textsc{sel} \cdot \left(\textsc{prep} \otimes I_n\right).
\end{equation}
The circuit diagram corresponding to this transformation is shown in Fig. \ref{fig:LCU}. 

\begin{figure}[ht]
    \centering
    \begin{tikzpicture}
    \hspace{-.3cm}
\node[scale=.98]{
    \begin{quantikz}
        \lstick{$\ket{0}^{\otimes q}$}& \gate{\textsc{prep}} &\qw \mathlarger{\mathlarger{\mathlarger{\mathlarger{\oslash}}}} \vqw{1} &\gate{\textsc{prep}^\dagger} &\qw\rstick{$\bra{0}^{\otimes q}$}\\
        \lstick{$\ket{\psi}$}& \qw  &\gate{\textsc{sel}} &\qw&\qw\rstick{$\frac{A}{\alpha}\ket{\psi}$} 
    \end{quantikz}
    };
    \end{tikzpicture}
    \caption{Circuit for linear combination of unitaries. $\textsc{prep}$ prepares the LCU coefficients $\beta_l$ into a quantum state $\sum_l \sqrt{\beta_l}\ket{l}_q$, and $\textsc{sel}$ selects the unitary operator to apply conditioned on the state of the ancilla register, where the $\oslash$ notation is to be interpreted as a control on all states in the ancilla register. After uncomputing the ancilla register with $\textsc{prep}^\dagger$, the normalized dipole operator is applied to the system register conditioned on the ancillary system with $q = O(\log(k))$ being measured in the all-zero state.}
    \label{fig:LCU}
\end{figure}
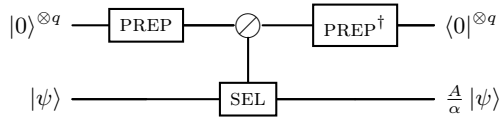

A block encoding is characterized by 3 parameters, the subnormalization factor $\alpha$, the number of ancilla qubits $q$, and the accuracy of the block encoding $\epsilon$.
\begin{defn}[$(\alpha, q, \epsilon)$ Block Encoding]
    For an $n$ qubit matrix $A$, we say that the $n+q$ qubit unitary matrix $U_A$ is an $(\alpha, q, \epsilon)$ block encoding of $A$ if
    \begin{equation}
        ||A - \alpha (\bra{0}^{\otimes q}\otimes I_n) U_A (\ket{0}^{\otimes q}\otimes I_n)|| \leq \epsilon
    \end{equation}
    If we can implement the above unitary $U_A$ exactly, then we call it an $(\alpha, q, \epsilon=0)$ block encoding of $A$.
    \label{def: Block Encoding}
\end{defn}

In the worst case, the complexity of implementing the state corresponding to the coefficients $\beta_l$ is exponentially hard in the number of qubits, i.e. $\tilde{O}(2^q)$ \cite{nielsenQuantumComputationQuantum2010} gates. However, since the number of qubits is logarithmic in the number of terms in the LCU, the worst case circuit complexity is $O(2^q) = O(k)$. If $k$ scales as $\text{poly}(N)$ then the state preparation can be done efficiently. However, in some cases such as the coefficients belonging to a log-concave distribution, the state preparation can be significantly sped up to $O(\polylog(N))$~\cite{grover2002creating}.  This circuit is applied twice in the \textsc{prep} and \textsc{prep}$^\dagger$ routines. 

Since this procedure is non-deterministic, there is a probability that the computation fails. Success is flagged by the ancilla system being in the state $\ket{0}^{\otimes q}$ and is directly related to the subnormalization factor $\alpha$. When we apply the block encoding to a state $\ket{\psi}_n\ket{0}_q$ we obtain
\begin{equation}
     \left(\frac{A}{\alpha}\ket{\psi}_n\right) \ket{0}_q + \ket{\perp}_{n+q}.
\end{equation}
where $\ket{\perp}$ denotes the orthogonal undesired subspaces. Therefore, the success probability of being in the correct subspace $\ket{0}_q$ is, 
\begin{equation}
    P_{\text{succ}} = \left| \left| \frac{A}{\alpha}\ket{\psi}_n \right| \right|^2.
\end{equation}
Therefore, the failure/success probability directly depends on the ratio $\xi/\alpha$ where
\begin{equation}
    \xi \equiv ||A\ket{\psi}||.
    \label{eq:xi}
\end{equation}
In the absence of known bounds for $\xi$, the best-case bound is that $\alpha$ upper bounds $\xi$ by a constant or logarithmic factor so that the success probability can be boosted to $O(1)$ with only $\widetilde{O}(1)$ rounds of robust amplitude amplification \cite{berryHamiltonianSimulationNearly2015a}. In general, we will need to perform $\widetilde{O}(\alpha/\xi)$ rounds of amplitude amplification to boost the success probability of the block encoding to be $O(1)$. 

\subsection{Quantum Signal Processing}
Quantum signal processing (QSP) provides a systematic procedure for implementing a class of polynomial transformations to block encoded matrices. 
For Hermitian matrices $A$, the action of a scalar function $f$ can be determined by the eigendecomposition of $A = UDU^\dagger$, as
\begin{equation}
    f(A) := U \sum_{i}f(\lambda_i)\ket{i}\bra{i}U^\dagger.
    \label{eq:mat func}
\end{equation}
where $\lambda_i$ is the $i$th eigenvalue, and $\ket{i}$ is the $i$th eigenvector. For the cases we consider here, the operators are Hermitian, and we use this characterization of matrix functions as a definition.
QSP provides a method for implementing polynomial transformations of block encoded matrices. To implement a degree $d$ polynomial, QSP \cite{lowOptimalHamiltonianSimulation2017, Lin2022} uses $O(d)$ applications of the block encoding and therefore its efficiency depends closely on the rate a given polynomial approximation converges to the function of interest and the circuit complexity to implement the block encoding. 

 QSP uses repeated applications of the block encoding circuit to implement polynomials of the block encoded matrix. We assume that $A$ is a Hermitian matrix that has been suitably subnormalized by some factor $\alpha$ so that $||A/\alpha||\leq 1$, and that we can express $A$ with its eigendecompostion. QSP exploits the surprising fact that the block encoding $U_A$ can be expressed as a direct sum over one and two dimensional invariant subspaces which correspond directly to the eigensystem of $A$. In turn this allows one to \textit{obliviously} implement polynomials of these eigenvalues in each subspace, which builds up the polynomial transformation of the block encoded matrix $A/\alpha$ overall. QSP is characterized by the following theorem, which is a rephrasing of Theorem 4 of Ref. \cite{lowOptimalHamiltonianSimulation2017}.
 
\begin{theorem}[Quantum Signal Processing (Theorem 7.21 \cite{Lin2022})]
    There exists a set of phase factors $\boldsymbol{\Phi} := (\phi_0, \ldots, \phi_{d}) \in \mathbb{R}^{d+1}$ such that
    \begin{equation}
    \begin{aligned}
        U_\phi(x) &= e^{i \phi_0 Z}\prod_{j=1}^d[O(x)e^{i\phi_j Z}]\\
        &=
        \begin{pmatrix}
            P(x) & - Q(x)\sqrt{1-x^2}\\
            Q^*(x)\sqrt{1-x^2} & P^*(x)
        \end{pmatrix}
    \end{aligned}
    \end{equation}
    where
    \begin{equation}
        O(x) = \begin{pmatrix}
            x & -\sqrt{1-x^2}\\
            \sqrt{1-x^2}& x
        \end{pmatrix}
    \end{equation}
    if and only if $P,Q \in \mathbb{C}[x]$ satisfy
    \begin{enumerate}
        \item $deg(P) \leq d, deg(Q) \leq d-1$
        \item $P$ has parity $d \mod 2$ and $Q$ has parity $d-1 \mod 2$, and
        \item $|P(x)|^2 + (1-x^2)|Q(x)|^2 =1 \hspace{.2cm}\forall x \in[-1,1]$
    \end{enumerate}
\end{theorem}
Given a scalar function $f(x)$, we can use QSP to approximately implement $f(A/\alpha)$ by using the block encoding of $A$, a polynomial approximation to the desired function, and a set of phase factors corresponding to the approximating polynomial. If $f$ is not of definite parity, we can use the technique of linear combination of block encodings to obtain a polynomial approximation to $f$ that is also of indefinite parity. This corresponds to implementing QSP for the even and odd parts respectively and then using an additional ancilla qubit to construct their linear combination. The last requirement, that $P(x) + (1-x^2)|Q(x)| = 1$ for every $x$, is the most restricting and requires normalization of the desired function which can be severe in some cases. However, the polynomial $P$ can be specified independent of the polynomial $Q$ so long as $|P(x)|\leq 1 \hspace{.2cm} \forall x \in[-1,1]$ and $P(x)\in \mathbb{R}[x]$ (See Corollary 5 of Ref. \cite{gilyenQuantumSingularValue2019a}), which is the case for the algorithm presented in this work.

Furthermore, there are efficient algorithms for computing the phase factors for very high degree polynomials.  Algorithms for finding phase factors for polynomials of degree $d = O(10^7)$ have been reported in the literature (see e.g. \cite{motlaghGeneralizedQuantumSignal2023, dongEfficientPhasefactorEvaluation2021}) and are surprisingly numerically stable even to such high degree. Given this result, conditions (1-3), and that a rapidly converging polynomial approximation exists, one can efficiently approximate smooth functions of the block encoded matrix $A/\alpha$ for very high degree polynomials. Furthermore, QSP provides an explicit circuit ansatz that allows one to know the exact circuit once the phase factors are found. This is a powerful tool since many computational tasks can be cast in terms of matrix functions.

\section{\label{sec:level3}Algorithm\protect }
In this section we present a quantum algorithm for computing the response and electronic absorption spectrum of a molecular system. The focus of our algorithm is in obtaining estimates of the induced dipole moments since these values give a direct way to compute absorption probabilities as described in the previous section. We desire to compute the response function by exploiting the sum-over-states form of the polarizability given by Eq.~\eqref{eq:offdiagpolar}. This requires finding the relevant overlaps corresponding to particular ground-to-excited state transitions. In this section we will use the convention $D^i$ to refer to the matrix representation of the dipole operator in direction $i$ to differentiate from its matrix elements $d_{i,(nm)} = \bra{\psi_m}D^i\ket{\psi_n}$. $D$ without indices refers to arbitrary dipole operator. The terms in the sum-over-states form require approximations to terms of the form
\begin{equation}
d^{(i,j)}_{0,n}=\bra{\psi_0}D^j\ket{\psi_n}\bra{\psi_n}D^i\ket{\psi_0},
\label{eq: off-diagonal terms}
\end{equation}
and is non-zero whenever the transition $0\rightarrow n$ is allowed by symmetry, commonly referred to as a \textit{selection rule}.

In the case of the second quantized chemistry Hamiltonian, we assume that spin-orbitals are mapped to qubits in a one-to-one fashion on a quantum computer using the Jordan-Wigner transformation in this work. In the Jordan-Wigner transform $n = N$ where each of $n$ qubits are mapped to one of the $N$ spin-orbitals in a second quantized representation. We will use $[a,b]$ to refer to a window of excitation frequencies over which we will wish to approximate the response on and $\Delta := b-a$ as width of the interval over which we wish to estimate the response. In this section, we will show how one can obtain systematically improvable approximations to frequency-dependent nth-order response properties. Combining these approximations, we obtain an approximation to the molecular response to the perturbing electric field over the total length of a spectral region of interest.

The algorithm is based off of a simple observation, given access to an eigenstate, say the ground state, $\ket{\psi_0}$ of the electronic Hamiltonian $H_0$, and a block-encoding of a potential, such as the dipole operator $D$, we can use Fermi's rule to approximate transition rates and strengths on the quantum computer. An algorithm implementing this procedure can be realized with a generalized Hadamard test, and uses only one ancilla qubit beyond those used for the block-encoding. 

Additionally, we define the indicator function
\begin{equation}
    1_{[a,b]}(x) = \begin{cases}
        1 & x \in [a,b]\\
        0 & \text{else}.
    \end{cases}
\end{equation}
To understand how this transformation produces the desired output, we will assume that we are able to exactly implement the indicator function of our matrix $1_{[a,b]}(H_0)$. In the later sections, we analyze the scaling of the polynomial for a desired error tolerance. When applied to a matrix, the indicator function serves as a projection onto the subspace of $H_0$ spanned by eigenvalues in $[a,b]$. By this we mean,
\begin{equation}
    1_{[a,b]}(H_0) \equiv \sum_{\lambda_j\in[a,b]}\ket{\psi_j}\bra{\psi_j},
\end{equation}
where $\sum_{\lambda_j \in [a,b]} \equiv \sum_{j: \lambda_j \in [a,b]}$ and where we have written $H_0$ as its eigendecomposition
\begin{equation}
    H_0 = \sum_j \lambda_j \ket{\psi_j}\bra{\psi_j}.
\end{equation}

Let $\ket{\psi_0}$ be the exact ground state of $H_0$ and consider the dipole operators $D$ and $D'$. Before we make precise the statement of our algorithm, consider the following sequence of operations
\begin{align*}
    D'1_{[a,b]}(H_0) D\ket{\psi_0} &= D'1_{[a,b]} \sum_{n>0}d_{0,n}\ket{\psi_n}\\
    &= D' \sum_{n>0}d_{0,n}\sum_{\lambda_j \in [a,b]} \ket{\psi_j}\delta_{nj}\\
    &= D' \sum_{\lambda_j \in [a,b]} d_{0,j} \ket{\psi_j}\\
    &=  \sum_{k\neq j}\sum_{\lambda_j \in [a,b]} d_{0,j} d'_{j,k}\ket{\psi_k}\\
\end{align*}
then, if we could compute the inner product of that state with the ground state $\ket{\psi_0}$, we would find obtain the desired terms
\begin{equation}
    \bra{\psi_0}D'1_{[a,b]}(H_0) D\ket{\psi_0} = \sum_{\lambda_j \in [a,b]}{d_{0,j} d'_{j,0}}.
\end{equation}

For $[a,b]$ small enough so that $|b-a| < \gamma$, the parameter controlling the desired spectral resolution, we can faithfully approximate the response by summing over the approximations we find on each small interval of size $o(\gamma)$. If $[a,b]$ contains only one excitation, then this approximation would be exact up to sampling errors, and the only uncertainty is the position in the window where the excitation occurred. In the next section we present an algorithm that can improve this resolution at the Heisenberg limit.  We introduce the following definitions, let 
\begin{equation}
\begin{aligned}
    d^{x,x'}_{[a,b]} &:= \sum_{\lambda_j \in [a,b]} d_{0,j}d'_{j,0}\\
    d^{x',x}_{[a,b]} &:= \sum_{\lambda_j \in [a,b]} d'_{0,j}d_{j,0}
\end{aligned}
\end{equation}
then since the polarizability
\begin{equation}
    \alpha_{x,x'}(\omega) \equiv \sum_{n\neq 0}\frac{d^{x,x'}_{\omega_n}}{\omega_{n0} +\omega+i\gamma}+ \frac{d^{x',x}_{\omega_n}}{\omega_{n0}-\omega -i\gamma},
    \label{eq: polar SOS}
\end{equation}
we can write an approximation to this sum in terms of the approximations for each sub-interval $[a,b]$ estimating $\omega_{n0}$ by the central frequency of $[a,b]$, $\widetilde{\omega}_{a,b} \equiv |b-a|/2$ 
\begin{equation}
    \widetilde{\alpha}_{x,x'}(\omega) = \sum_{i}\sum_{\omega_{n0} \in [a_i,b_i]}\frac{d^{x,x'}_{[a_i,b_i]}}{\widetilde{\omega}_{a_i,b_i} +\omega+i\gamma}+ \frac{d^{x',x}_{[a_i,b_i]}}{\widetilde{\omega}_{a_i,b_i}-\omega -i\gamma},
    \label{eq: polar appx}
\end{equation}
with $\cup_{i}[a_i,b_i] \subset [-\alpha,\alpha]$.
In the next section we construct a quantum algorithm to obtain estimates to the factors $d_{[a,b]}^{x,x'}$ which we can then use in conjunction with Eq. ~\eqref{eq: polar appx}  to obtain an approximation to ~\eqref{eq: polar SOS}. 

\subsection{Full Algorithm and Cost Estimate}
Our algorithm will rely on block-encodings of the Hamiltonian $H_0$ and dipole operators $D$ and $D'$, which will be accessed through their $(\alpha, m, 0)$ and $(\beta, k, 0)$ block-encodings $U_H$ and $U_D, U_{D'}$ respectively. We will additionally need access to a high-precision approximation to the ground state. This assumption can be made more precise by first assuming that the initial state $\ket{\psi}$ has been prepared with squared overlap $p_0 = \left\vert\braket{\psi_0}{\psi}\right\vert^2$ and that the spectral gap $G = \lambda_1-\lambda_0$, with neither being exponentially small. Under these assumptions, the eigenstate filtering method proposed in \cite{tongQuantumEigenstateFiltering2022,linOptimalPolynomialBased2020a}, constructs an algorithm which can prepare the ground state with
$O\left(\frac{\alpha}{\sqrt{p_0} G}\log(1/\epsilon)\right)$
queries to $U_H$. 

Given these constraints, we can apply this algorithm to efficiently transform some good initial state $\ket{\psi}$ into a high fidelity approximation of the true ground state $\ket{\psi_0}$. Once the high fidelity ground state is in hand, we can additionally use eigenstate filtering to obtain a high-fidelity estimate to the ground energy as well. Although the assumption of having an exact ground state is unlikely to be achieved in practice, this algorithm allows us to transform our initial state with $1/\text{poly}(n)$ overlap, to a state with arbitrarily close-to-1 overlap with the true ground state. We consider this to be a preprocessing step, and we assume for the duration of this work access to exact ground state preparation. 

To motivate the construction, consider the operation $\mathcal{U}$, which we define as the product of block-encodings,
\begin{equation}
    \mathcal{U} = U_{D'}U_{\tilde{1}_{[a,b]}(H)}U_{D}.
\end{equation}
Forming the linear combination $\frac{1}{2}(\mathcal{U} + I)$ implements the Hadamard test for non-unitary matrices and when applied to the ground state $\ket{\psi_0}$, we observe that the success probability of measuring the ancilla qubit to be $\ket{0}$ in the LCU is
\begin{equation}
    P(0) = \left|\left|\frac{1}{2}(\mathcal{U} + I)\ket{\psi_0}\right|\right|^2
\end{equation}
which can be simplified to
\begin{equation}
    P(0) = \frac{1}{2}\left(1 + \frac{d^{x,x'}_{[a,b]}}{\zeta_{\omega}}\right).
\end{equation}
where 
\begin{equation}
    \zeta_{\omega} = \left(\alpha + 1+\omega\right)\beta^2
\end{equation}
is the subnormalization factor for this block encoding.
In the following section, we bound the cost of performing these operations and the number of repetitions needed to obtain accurate approximations to the response to the desired spectral resolution $o(\gamma)$ in terms of these parameters.

\subsubsection{Block-encoding $H_0$ and Dipole Operators}
Our algorithm relies on access to block-encodings of the dipole operators, which we denote $D$ and $D'$, and $H_0$. We bound the cost to implement these operations and show that it is subdominant to the cost of block-encoding $H_0$.  Block-encoding the electronic structure Hamiltonian has been discussed in great detail in the literature, so we will not explicitly analyze the block-encoding gate complexity for $H_0$ and instead bound the number of queries to the block-encoding $U_H$ given by the state-of-the-art methods for electronic structure. 

One approach to block-encoding the dipole operator $D$ is to start from its second quantization. We can use the Jordan-Wigner transformation to write $D$ as a linear combination of Pauli matrices
\begin{equation}
    D = \sum_j^{O(N^2)} c_j P_j,
\end{equation}
with $P_j \in \{I, X, Y, Z\}^{\otimes N}$.
From this we immediately obtain a representation of $D$ as linear combination of $O(N^2)$ unitary matrices. Accordingly, the dipole operator can be block-encoded using $k = \ceil{\log(N^2))}$ ancilla qubits and $O(N^2)$ controlled Pauli operators. The subnormalization factor $\beta$ is the $1-$norm of the coefficients in this decomposition and is therefore
\begin{equation}
    \beta = \sum_{j=0}^{O(N^2)} |c_j|.
    \label{eq:alpha}
\end{equation} 
We assume each of the coefficients in the LCU is a constant with respect to the number of basis functions $N$, therefore we have that $\beta \in O(N^2)$.

An alternative approach to block-encoding the single particle Fermionic operator is to use a single-particle basis transformation $U$ \cite{babbushLowDepthQuantum2018} to represent the creation and annihilation operators in the single-particle basis that makes $D$ diagonal. That is
\begin{equation}
    \sum_{p,q}d_{pq}a_p^\dagger a_q = U\left(\sum_{j} \lambda_j a_j^\dagger a_j\right)U^\dagger,
\end{equation}
where the eigendecomposition of $d=u\cdot \mathrm{diag}[\lambda_0,\cdots,\lambda_{N-1}]\cdot u^\dagger$. 
This immediately leads to a block-encoding of the dipole operator with $\beta = N ||d||$, where $||d||=\max_j|\lambda_j|$. However, applying the result of Ref. \cite{tongFastInversionPreconditioned2021a} Appendix J, this can be further improved. Let $\vec{|\lambda|}$ sort the $\lambda_j$ in magnitude from largest to smallest. Then
\begin{equation}
    \beta = \sum_{j=0,\cdots,\eta-1}\vec{|\lambda|}_j=\mathcal{O}(\eta ||d||)
\end{equation}
Importantly, any single-particle rotation can be implemented using only $\mathcal{O}(N^2)$ quantum gates. 

When the block-encoding of the dipole operator is applied to the molecular ground state, the state of the quantum computer is,
\begin{equation}
    \frac{D}{\beta}\ket{\psi_0}\ket{0}_k + \ket{\perp}.
\end{equation}
Upon measurement of the ancilla register returning the all-zero state, we have knowledge that the state of the system register must have been
\begin{align}
    \frac{D}{\beta}\ket{\psi_0} = \sum_{j>0}\frac{d_{0,j}}{\beta}\ket{\psi_j}.
    \label{eq:afterLCU}
\end{align}
Where we have dropped the $k$-qubit ancillary register used in the block-encoding since we have assumed that register has been uncomputed and measured in the $\ket{0}_k$ state.

At this stage of the algorithm we have prepared a superposition of excited states with amplitudes proportional to the probability of the respective ground to excited state transition. In the next step we show how to extract response properties at particular energies using an eigenstate filtering strategy.

\subsubsection{Block-encoding the indicator function}
The indicator function can be implemented using a suitable polynomial approximation and quantum signal processing of the block-encoded matrix $U_H$. In the context of quantum algorithms, the indicator function was has been used as part of the eigenstate filtering subroutine explored in Ref. ~\cite{linOptimalPolynomialBased2020a}, which constructs a polynomial approximation to the Heaviside step function smoothed by a mollifier. Other work such as Refs. ~\cite{gilyenQuantumSingularValue2019a,lowHamiltonianSimulationUniform2017}, provides an improvement to Ref.~\cite{linOptimalPolynomialBased2020a} by a factor of $O(\log(\delta))$ using a polynomial approximation to the error function $\erf(kx)$. This factor is non-trivial, as in many cases the subnormalization factor rescales $\delta\in O(1)$ to $\delta/\alpha \in O(1/\alpha)$. Lemma 29 of Ref.~\cite{gilyenQuantumSingularValue2019a} promises the existence of a polynomial approximation to the symmetric indicator function that can be constructed using a $d=O(\delta^{-1}\log(\epsilon^{-1}))$ degree polynomial where $\delta\in (0,1/2)$. 

We would like to block encode the indicator function over the symmetric interval $\left[\omega-\Delta/2,\omega +\Delta/2\right]$. This corresponds to choosing some fixed central frequency $\omega$ and forming the interval of width $\Delta$. One straightforward way to implement the shift is to simply perform an LCU of the original Hamiltonian with the identity matrix of the same dimension. Since we would like to study the excitation energies, we will assume that we have shifted the Hamiltonian so that $\lambda_0 = 0$. We would like to implement the linear combination $H_0 - \omega I$ for some $\omega > 0$. We can construct a block-encoding of the shifted Hamiltonian with an additional qubit as in Fig. \ref{fig:ShiftedHamCircuit}. 
\begin{figure}[h]
    \centering
    \begin{quantikz}
        \lstick{$\ket{0}$}&\gate{R_y(\theta)}&\octrl{1}&\gate{Z}&\gate{R_y^\dagger(\theta)}&\qw\rstick{$\bra{0}$}\\
        \lstick{$\ket{0}$}&\qwbundle{m}&\gate[2]{U_H}&\qw\\
        \lstick{$\ket{\psi_0}$}&\qw&\qw&\qw
    \end{quantikz}
    \\
    =
    \begin{quantikz}
        \lstick{$\ket{0}$}&\qwbundle{m+1}&\gate[2]{U_{H_0-\omega I}}&\qw\\
        \lstick{$\ket{\psi_0}$}&\qw&\qw&\qw
    \end{quantikz}
    \caption{Quantum circuit for implementing $(\alpha + {\omega}, m+1,0)$ block-encoding of $H_0 - \omega I$, which is just the linear combination of $H$ with the identity, with $\alpha = 1 + \omega$. The prepare oracle can be implemented by the rotation $R_y(\theta)\ket{0} = \cos(\theta/2)\ket{0} + \sin(\theta/2)\ket{1}$ and $\theta = 2\arccos\left(\left(1+\omega^2\right)^{-1/2}\right)$, so that $R_y(\theta)\ket{0} =\frac{1}{\sqrt{1+\omega^2}} \left(\ket{0} + \sqrt{\omega}\ket{1}\right)$.}
    \label{fig:ShiftedHamCircuit}
\end{figure}
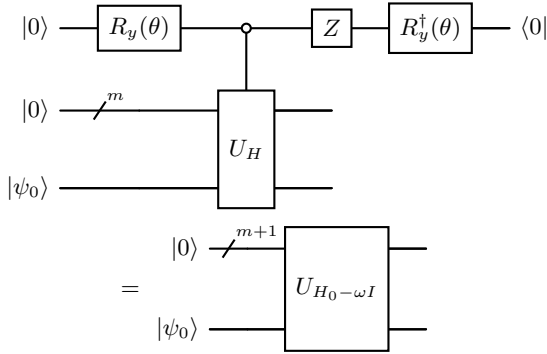

 We use the notation $\widetilde{1}_{[a,b]}$ to refer to the approximated  indicator function on some interval $[a,b]$. Block encoding the indicator function can be done with an additional $O(d)$ queries to the block-encoding of $H_0 - \omega I$ to implement the approximate indicator function applied to the shifted Hamiltonian. That is, given access to an $(\alpha + \omega, m, 0)$ block-encoding of $H_0$, we can construct an $(\alpha + \omega, m+1, 0)$ block-encoding of $H_0 - \omega I$. Then, using QSP we can form an $(\alpha + \omega, m+2, \epsilon)$ block-encoding of $\widetilde{1}_{[\omega-\Delta,\omega+\Delta]}(H_0)$. 

The next step of the algorithm is to apply another dipole operator to the state following the application of the filtering function, this step follows the same procedure as the block-encoding of the dipole operator we performed previously. In the following section, we apply these steps in a controlled fashion with using a small number of additional ancilla qubits using the non-unitary Hadamard test to estimate the desired terms.

\subsubsection{Non-unitary Hadamard test}
In order to obtain the terms relevant to computing response functions in the sum-over-states form, we need to compute the inner products $d^{x,x'}_{[a,b]} = \bra{\psi_0}D' 1_{[a,b]}(H_0) D\ket{\psi_0}$. The standard way to estimate inner products of unitary matrices on a quantum computer is the Hadamard test. Using the same circuit, but instead replacing the unitary matrix with a block-encoding renders an algorithm for estimating inner products of non-unitary matrices. This idea has been explored previously, and an implementation of this algorithm can be found in appendix D of Ref. \cite{tongFastInversionPreconditioned2021a}. 

For clarity, in this section we assume access to the operations above and that the indicator function is approximated well enough that we can ignore the error introduced by the approximation error. We will analyze the effect of this approximation on the desired expectation values in the following section. We define the map $\mathcal{U}$ according to its action on the ground state and ancilla registers used in the block-encoding as
 \begin{equation}
     \mathcal{U}\ket{\psi_0}\ket{0}_{m+k+k'} = \frac{D' 1_{[a,b]}(H_0) D}{\zeta_{\omega}}\ket{\psi_0}\ket{0}_{m+k+k'}  + \ket{\perp}.
 \end{equation}
Where we recall that $m$, $k$, and $k'$ are the ancilla qubits used for block encoding $H_0$ $D$ and $D'$ respectively. Now, we apply $\mathcal{U}$ in a controlled manner with an additional ancilla qubit using the Hadamard test circuit. Measuring this ancilla register returns the $\ket{0}$ state with probability:
\begin{equation}
    P(0) = \frac{1}{2}\left(1+ \frac{1}{\zeta_{\omega}}\bra{\psi_0}D' 1_{[a,b]}(H_0) D\ket{\psi_0}\right).
\end{equation}
It is easy to show that $\bra{\psi_0}D' 1_{[a,b]}(H_0) D\ket{\psi_0}$ gives the desired overlaps
\begin{align*}
   \bra{\psi_0}D' 1_{[a,b]}(H_0) D\ket{\psi_0}=\sum_{\omega_j \in [a,b], j\neq 0}d_{0,j}d'_{j,0} \equiv d_{[a,b]}^{x,x'}
\end{align*}

We can see that $\mathcal{U}$ corresponds to a product of the block-encodings we formed in the previous section, therefore it can be realized via the $(\zeta_{\omega}, m + 2 + k + k', 0)$ block-encoding formed by the product of block-encoded matrices \cite{gilyenQuantumSingularValue2019a}. One should note that in the case where $D' = D$, one can instead do the Hadamard test on only the indicator function to approximate the inner product $\bra{\psi_0}D^\dagger 1_{[a,b]}(H_0)D\ket{\psi_0}$. This gives an improvement by a factor of $O(\beta)$ in the query complexity over the off-diagonal case.

This introduces the main procedure for calculating the terms in the sum-over-states representation of response functions and a graphical representation of this procedure is provided in Fig \ref{fig:Spec Diag}.  In the next section, we show how to systematically improve the resolution by reducing the size of the window $[a,b]$ and determine the response in bins to high resolution. Additionally, we show how these subroutines can be iteratively applied to efficiently approximate response functions to any order. 

\begin{figure*}[ht]
    \centering
    \includegraphics[width=.7\paperwidth]{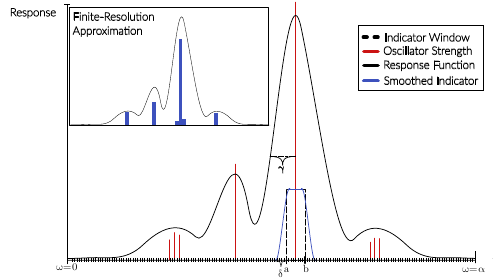}
    \caption{This schematic is a visualization of one run of our algorithm, assuming only first order (linear) response. The x-axis is the frequency, $\omega(\lambda) = \lambda - \lambda_0$, the y-axis the measured response. The goal is to approximate the response in the window $[a,b]$, by forming the indicator function $1_{[a,b]}$ over the portion of the $[\omega(a),\omega(b)]$ using the the block-encoding of $H_0$ and QSP. The smooth line in black is the damped response function given by a Lorentzian broadening parameter $\gamma$, the red sticks are the oscillator strengths associated with particular transitions. We approximate the response by convolving the smoothed indicator function (in blue) with the signal (in red) to estimate the strength of the response over the region indicated $a$ and $b$. The region denoted $\delta$ corresponds to the smoothing parameter used in approximating the indicator function.  Using this approach over different regions with smaller and smaller windows, we obtain a systematically improvable estimate of the transition strength associated with a particular transition energy window, as shown in the inset.}
    \label{fig:Spec Diag}
\end{figure*}

\subsubsection{Approximating first-order polarizability tensor}
Although we have a quantum state whose amplitudes correspond to the dipole transition probabilities in a particular window, we would like to use these estimates to compute response functions. The properties of the transition dipole will affect the difficulty of this task. In the most pessimistic case, the dipole prepares a near uniform superposition over the exponentially large set of excited states. In this case it can be very inefficient to extract high precision estimates of spectrally localized properties of the response function, as each amplitude is exponentially suppressed leading to small success probabilities. However, this behavior is unlikely for many systems of interest as it would imply that electronic transitions are nearly equiprobable for every incident frequency $\omega$. Therefore, in practice we expect some concentration of the response function to relatively small portions of the spectrum similar to the schematic in Fig~\ref{fig:Spec Diag}.

The algorithm in the above section provides a procedure for approximating the dipole response in the region $[a,b]$ to compute $d^{x,x'}_{[a,b]}$ and $d^{x',x}_{[a,b]}$ and provides an approximate excitation energy for those responses as the center of the window $\widetilde{\omega}_{a,b} = |b-a|/2$. To approximate response functions such as the linear polarizability, we need to approximate all of the non-zero terms in the sum-over-states formalism. An inspection of this definition of the polarizability reveals that terms with large dipole response will dominate the behavior of the function.

These terms need to be evaluated to high precision and the window $[a,b]$ that we use to approximate this quantity needs to be smaller than some desired spectral resolution $\gamma$. For fixed $\gamma$, in order to adequately resolve the response function requires the window $[a,b]$ to satisfy $\Delta \equiv |b-a| < \gamma$. Since by definition $\delta < \Delta < \gamma$, and the subnormalization $\alpha$ rescales this interval, this immediately implies that $\delta = o(\alpha^{-1})$ so that the degree of polynomial needed is $\widetilde{O}(\alpha)$. This increases the number of queries to $U_{H}$ by an additional factor of $\widetilde{O}(\alpha)$. Therefore it is prudent to postpone this cost until it is absolutely necessary.  To more rapidly approximate the frequency dependent response, we first introduce an algorithm to determine more pronounced regions in the response and resolve them to higher precision. This algorithm, which uses ideas based on a modified binary search similar to that found in Ref. \cite{linOptimalPolynomialBased2020a}, can be implemented as a \textit{linear combination of Hadamard tests}.

For simplicity, assume that the Hamiltonian has already been rescaled so that its spectrum lies in $[-1,1]$. Take $B, |B| \in O(1)$ to be a set of bins and for each $b_i \in B$ where $b_i = \left[-1 +\frac{2i}{|B|},-1 +\frac{2(i+1)}{|B|}\right]$. The problem then becomes determining which bins have the largest response. Since our goal is not to approximate $d^{x,x'}_{b_i}$ for each $b_i$, but instead to determine inequalities between the $d^{x,x'}_{b_i}$'s, the number of samples can be greatly decreased. Since the determination of the inequality $d^{x,x'}_{b_i} > d^{x,x'}_{b_j}$ is made by sampling, we only need to ensure this inequality holds with constant probability. The number of samples $N_s$ needed to determine that $d^{x,x'}_{b_i} > d^{x,x'}_{b_j}$ with high probability is only $\widetilde{O}(1)$ as a result of a Chernoff bound (See Appendix \ref{app: Filtering} Lemma \ref{lem: inequality test}). 

Define $\mathcal{U}_{b_i} =U_{D'}U_{1_{b_i}}(H_0)U_D$, which is related to the average of the response on region $b_i$. For each of the $\mathcal{U}_{b_i}$ we will append an additional ancilla qubit so that we can form the LCU $\frac{1}{2}(\mathcal{U}_{b_i} + I) =: \mathcal{U}_i$. Next we add an ancilla register of $r = \log(|B|)$ qubits and apply the diffusion operator on the ancilla register. Then, we apply each of the $\mathcal{U}_i$ to the ground state controlled on the ancilla register. The result of these operations prepares the state
\begin{equation}
    \frac{1}{\sqrt{2^{r+1}}}\sum_{i=0}^{2^{r}-1}\ket{i}\mathcal{U}_i\ket{\psi_0}\ket{0}_{|B|+a},
\end{equation}
using an additional $|B| + \log(|B|)$ qubits with $a= m+k+k'+2$ the ancilla qubits used in block encoding $\mathcal{U}_{b_i}$.

Observe that the probability of observing the $i$th bitstring of the ancilla register can be expressed as 
\begin{equation}
    P(i) = \frac{1}{2|B|}\left(1+ \text{Re}\left(\frac{d^{x,x'}_{b_i}}{\zeta_{\omega}}\right)\right),
\end{equation}
which resulted from the fact that the amplitude in front $\ket{i}$ is directly related to the Hadamard test success probability. Therefore, in the same way as the standard Hadamard test, the frequency at which we observe the ancilla register in state $\ket{j}$ upon measurement is directly related to the strength of the response in the region $b_j$. 

Under reasonable assumptions, the number of samples to determine the desired inequalities to high probability is $O(\log(1/(1-P))|B|^2)$ times, where $P$ is the probability that the ordering obtained by a finite number of samples is correct. Since by assumption the number of bins $|B|$ is only $O(1)$ the $\delta$ parameter for the approximate indicator function can be taken to be some small constant value less than $1/|B|$. Once we have sampled from this distribution we will have a rough sketch of which of the $B$ bins has the largest response.  Additionally, since the number of bins in this stage of the algorithm is independent of problem size, each one of these bins can be resolved using only $\widetilde{O}(1)$ queries to $U_H$.

Let $b_i$ be some bin which contained a large response and $|b_i|\in O(1/|B|)$. Our goal is to further resolve which subset(s) of $b_i$ have largest response. As above, we can subdivide $b_i$ into $|B|$ sub-intervals and once again perform the linear combination of Hadamard tests. Then by measuring the ancilla used in the LCU, we obtain information about which sub-bins of $b_i$ have large response. Continuing in this fashion we can iteratively resolve the response onto an interval of size $o(\gamma/\alpha)$. Pseudo code describing the control flow for one run of this algorithm can be found in Alg \ref{alg: inequality testing} and a visualization for how multiple runs of this algorithm might proceed can be found in Fig. \ref{fig:schematic} in Appendix \ref{app: Filtering}.  We detail the cost of determining this interval in the following theorem.

\begin{theorem}[Frequency bin sorting (informal)]
    There exists a quantum algorithm to determine frequency bins corresponding to large amplitudes to resolution $O(\gamma/\alpha)$ using $\widetilde{O}(\alpha^2\beta^2/\gamma)$ queries to the $(\alpha, m, 0)$ block-encoding of $U_H$ and $O(\alpha\beta^2)$ queries to the $(\beta, k, 0)$ block-encodings of $U_D$ and $U_D'$. (See Appendix \ref{app: Filtering} Theorem \ref{thm: search algorithm conv} for additional details)
    \label{thm: bin finding 1d}
\end{theorem}
\begin{proof}
        Let $\delta_0 = O(1)$ be the smoothing parameter for the Heaviside approximation and $[a,b]$ the bin we wish to find the largest response in with $\Delta \in O(1)$. Let $|B|$ be the number of sub-intervals we use at each stage of the algorithm. Let $B$ be an equal subdivision of $[a,b]$ into $|B|$ bins and let $b_i = [a+ib/|B|,a+(i+1)b/|B|] \in B$ be an interval of width $O(1/|B|)$. Define $\mathcal{U}_{i} = \frac{1}{2}\left(I + U_{D'}U_{\tilde{1}_{b_i}}U_{D}\right)$. 
        
        To prepare the linear combination of Hadamard tests with $\mathcal{U}_{i}$ requires constructing a linear combination of $O(|B|)$ terms of the form
        \begin{equation}
            \frac{1}{2|B|}\sum_{i}\mathcal{U}_i \otimes \ket{i}\bra{i},
        \end{equation}
        which requires $O(|B|)$ ancilla qubits for the linear combination with the identity to form the Hadamard test and an additional $\log(|B|)$ ancilla for the linear combination of Hadamard tests. This in turn queries the block-encodings of the form $\mathcal{U}_i$ $O(|B|)$ times.
        
        Since splitting into $|B|$ intervals requires that $\delta_0 \rightarrow \delta_0/|B| =: \delta_1$, we have that the number of queries to the block-encoding $U_{H_0}$ increases by a factor of $\widetilde{O}(|B|)$ by Lemma 29 of Ref.~\cite{gilyenQuantumSingularValue2019a}. However, since we only need to find the sub interval of $b_1$ with large response to some constant probability, we only need to repeatedly measure the ancilla $\widetilde{O}(1)$ times. The first iteration therefore requires $\widetilde{O}(|B|^2)$ queries to $U_H$. Therefore, the cost is entirely dominated by the size and number of the intervals.
        
        If we let $b_2$ be a subdivision of $b_1$ into $|B|$ bins, we have $\delta_0 \rightarrow \delta_0/|B|^2$, increasing the number of queries to $O(|B|^3)$. In general, if we repeat this algorithm for $k$ iterations, we find that the number of queries to $U_H = O(|B|^2 + |B|^3 + \cdots |B|^k) = O(|B|^k)$. We desire that $|B|^{-k} < \gamma$, meaning that the bins at iteration $k$ are smaller than the desired resolution $\gamma$. This implies $k = O(\log(1/\gamma))$. For the simplest case of $|B|=2$ requires only a single additional ancilla, and the total number of queries to $U_H$ is $O(2^{k+1}) = 2^{2+\log(1/\gamma)}$ this corresponds to a total number of queries to $U_H$ that is $O(1/\gamma)$, which upon rescaling by the subnormalization gives a number of queries that is $O(\alpha/\gamma)$. 
        
        However, because there is a success probability associated with the block-encodings, this increases the query complexity by a factor of $\widetilde{O}(\alpha \beta^2)$ to obtain success with some constant probability. This results in the final query complexity $\widetilde{O}\left(\alpha^2\beta^2/\gamma\right)$ queries to $U_H$ and $\widetilde{O}(\alpha \beta^2)$ queries to $U_D$ and $U_{D'}$. Since we are able to improve the spectral resolution at a rate which is linear in $\gamma^{-1}$, the desired resolution, this algorithm saturates the Heisenberg limit up to scalar block encoding factors.
\end{proof}

This procedure allows us to systematically ``sift" through the spectrum when searching for regions where the response is most prominent. This allows us to focus the computational resources in approximating the response to high-accuracy only in regions where the response is most pronounced. To ensure a frequency resolution consistent with the problem, the spectrum should be resolved to $o(\gamma)$, which is the broadening parameter. Therefore, in regions where the response is largest we continue to halve the domain until the size of the bins are $O(\gamma)$.  Once again, we refer the reader to Fig. \ref{fig:schematic} for a visualization.

The main source of cost in the polynomial approximation occurs from the size of the smoothing parameter $\delta$. As a result, iterations with larger bins have smaller cost than iterations with smaller bins. Furthermore, since the goal is only to determine inequalities amongst the bins, and not their values, the bins can be sorted with high probability using tail bounds of the sort found in Lemma \ref{lem: inequality test} in appendix \ref{app: Filtering}. Since this approach saturates the quantum limit asymptotically, we believe this approach is nearly optimal for block-encoding based algorithms of this kind.

Once we have run this procedure and have determined the most prominent frequency bins, e.g. at the location of a peak, we would like to accurately approximate the height of the peak. Chebyshev's inequality bounds the number of expected samples to achieve an $\epsilon$ accurate approximation scales to $N_{\text{samples}} \in O\left(\epsilon^{-2}\right)$ to achieve an $\epsilon$ accurate result. We would like to have that our algorithm saturates the Heisenberg limit for estimating the heights of the found peaks. However, this cannot be done with the standard Hadamard test. In order to obtain the Heisenberg limit for our algorithm we will need to use amplitude estimation which uses quantum phase estimation on the circuit corresponding to the block-encoding $\mathcal{U}_i$.

\begin{theorem}[Approximating bin heights (informal)]
    Given a small interval $[a,b]$ with $\Delta \equiv |b-a|/2 = o(\gamma/\alpha)$ and a fixed $\epsilon > 0$, there exists a quantum algorithm that makes $O\left(\alpha^2\beta^2/(\gamma\epsilon)\right)$ queries to the $(\alpha, m, 0)$ block-encoding of $U_H$ and $O\left(\alpha\beta^2/\epsilon\right)$ queries to the $(\beta, m', 0)$ block-encodings of $U_D$ and $U_D'$ to obtain the average amplitude on the region $[a,b]$ to additive accuracy $\epsilon$. (See Appendix \ref{app: Filtering} Corollary \ref{cor:cor bounding response fn error} for additional details).
\end{theorem}
\begin{proof}
    Our goal is to estimate $\mu_{[a,b]} = \left\vert\bra{\psi_0}D'1_{[a,b]}(H)D\ket{\psi_0}\right\vert$, with an approximator $\hat{\mu}_{[a,b]}$ so that $|\hat{\mu}_{[a,b]} - \mu_{[a,b]}|\leq\epsilon$.

    Assume that,
    \begin{equation}
         \widetilde{\mu}_{[a,b]} = \left\vert\bra{\psi_0}D'\tilde{1}_{[a,b]}(H)D\ket{\psi_0}\right\vert,
    \end{equation}
    where $\widetilde{\mu}_{[a,b]}$ is the exact expected value of the approximate distribution with respect to the ground state wavefunction.

    We would like to find conditions on $\epsilon', \widetilde{\epsilon}$ so that whenever $|\widetilde{\mu}_{[a,b]} - \mu_{[a,b]}| \leq \epsilon'$ (error from approximation of indicator) and $|\widetilde{\mu}_{[a,b]} - \hat{\mu}_{[a,b]}| \leq \widetilde{\epsilon}$ (error from sampling compared true average of approximate indicator) $\implies |\mu_{[a,b]} - \hat{\mu}_{[a,b]}| < \epsilon$, our desired accuracy. Observe that
    \begin{align*}
        |\mu_{[a,b]} - \hat{\mu}_{[a,b]}| &= |\mu_{[a,b]}-\widetilde{\mu}_{[a,b]} + \widetilde{\mu}_{[a,b]} - \hat{\mu}_{[a,b]}|\\
        &\leq \left\vert\mu_{[a,b]}-\widetilde{\mu}_{[a,b]}\right\vert + \left\vert\widetilde{\mu}_{[a,b]} - \hat{\mu}_{[a,b]}\right\vert\\
        &\leq \epsilon' + \widetilde{\epsilon}
    \end{align*}
    which we desire to be smaller than $\epsilon$ which implies that $\epsilon' + \widetilde{\epsilon} < \epsilon$. Thus, whenever the approximation error $\widetilde{\epsilon}$ satisfies that $\tilde{\epsilon} \ll \epsilon$, we need to find an $\epsilon' = O(\epsilon)$ accurate approximation to result. However, because the success amplitude is rescaled by the subnormalization factor $\zeta_{\omega} = O(\alpha\beta^2)$, the success probability is $O\left((\alpha\beta^2)^{-2}\right)$, so the number of repetitions increases by a factor 
    of $O\left((\alpha\beta^2)^{2}\right)$.
    
    Using amplitude amplification improves this query complexity by a quadratic factor to $O(\alpha\beta^2)$. Then, using the block-encoding $\mathcal{U}$ with the quantum phase estimation algorithm to use amplitude estimation requires an additional factor of $O(1/\epsilon)$ controlled-applications of $\mathcal{U}$ to calculate the expectation value to accuracy $\epsilon$ requires $O(\alpha\beta^2/\epsilon)$ total queries to $\mathcal{U}$ and its controlled version. The algorithm which directly samples the output distribution would require the standard $O(1/\epsilon^2)$ repetitions, so the amplitude estimation routine improves over direct sampling by a factor of $O(1/\epsilon)$.

    The number of queries to form the block-encoding of the indicator function depends on the degree $d$ of the polynomial needed to approximate the indicator function to the desired precision. As was shown in Lemma 29 of \cite{gilyenQuantumSingularValue2019a}, we know that $d = \widetilde{O}(1/\delta) = \widetilde{O}(\alpha/\gamma)$, so the total number of queries to the block-encoding of $H_0$, $U_H$, is $\widetilde{O}(\alpha^2\beta^2/\gamma\epsilon)$. 
\end{proof}

At this point, we can estimate the cost to run this algorithm for the problem of second quantized quantum chemistry. We wish to derive a bound to the number of queries to $U_H$ in terms of the number of spin-orbitals $N$ and the number of electrons $\eta$. The naive bound on the subnormalization for $U_H$ would be $O(N^4)$ due to the $O(N^4)$ terms in the Coulomb Hamiltonian. However, this bound can be quite pessimistic in practice, and more efficient techniques such as double factorization \cite{pengLowrankFactorizationElectron2017, reiherElucidatingReactionMechanisms2017, babbushLowDepthQuantum2018} or the method of tensor hypercontraction \cite{leeEvenMoreEfficient2021} can reduce the number of terms significantly. In the low-rank double-factorization based approach, empirical results have shown that the number of terms in $H_0$ can be reduced to $O(N^2)$ while maintaining the desired accuracy (c.f. \cite{reiherElucidatingReactionMechanisms2017, Motta2021} Supporting Information sec VII.C.5). 

If we assume the empirical scaling, $\alpha \in O(N^2)$, and use the estimate that $\beta = O(||d|| \eta) \in O(N\eta)$ from the second quantization block-encoding procedure, we have that the total number of queries to $U_H$ is
\begin{equation}
    O\left(\frac{N^6\eta^2}{{\gamma}\epsilon}\log\left(\frac{1}{\epsilon}\right)\right),
\end{equation}
or, suppressing logarithmic factors we have 
\begin{equation}
    \widetilde{O}\left(\frac{N^{6}\eta^2}{{\gamma}\epsilon}\right).
\end{equation}

\subsection{Obtaining higher order polarizabilities, \label{subsec:higherOrd}}
The above algorithm provides a method for computing the overlaps and excitation energies which can then be used with a classical computer to obtain first-order response functions. A straightforward generalization of the above procedure can be used to compute various terms in the numerator as well as the energies needed to evaluate the denominator. 

Consider a product of block-encodings of the form 
\begin{equation}
    \mathcal{U}_2 = U_{D''}U_{\tilde{1}{[c,d]}}U_{D'}U_{\tilde{1}_{[a,b]}}U_{D}.
\end{equation} 
For now, ignoring subnormalization factors and ancillary qubits, if we can successfully implement controlled $\mathcal{U}_2$ in a Hadamard test, then we can gain information about$\bra{\psi_0}\mathcal{U}_2\ket{\psi_0}$. We can relate this to the desired information by observing
\begin{align*}
    &\bra{\psi_0}\mathcal{U}_2\ket{\psi_0} \sim \bra{\psi_0}D''1_{[c,d]}(H)D'1_{[a,b]}(H)D\ket{\psi_0}\\
    &= \sum_{\omega_{nm} \in [a,b]}\sum_{\omega_{m0} \in [c,d]}d_{0,n} d'_{n,m}\sum_{k}d''_{m,k}\braket{\psi_0}{\psi_k}\\
    &= \sum_{\omega_{nm} \in [a,b]}\sum_{\omega_{m0} \in [c,d]}d_{0,n} d'_{n,m}d''_{m,0}.
\end{align*}
Therefore, estimates this quantity provide estimates to $\omega_{m0} \in [a,b]$ and $\omega_{nm}\in [c,d]$ as well as the response observed in the region. In turn, this can be used to approximate the first hyperpolarizability.

The same procedure for roughly determining the distribution by forming a superposition over large regions of the frequency domain works for nonlinear response as well in the form of $n$-dimensional grid. First we discretize the interval $[-1,1]$ into $|B|$ bins with $b_i$ as above. Consider the Cartesian product $B\times B$ which is the set of all tuples $(b_i,b_j)$ for $b_i, b_j \in B$ and take two registers of size $\ceil{\log(|B|)}$ and prepare each in the uniform superposition and implement the controlled unitary
\begin{equation}
    \sum_{ij}\ket{ij}\bra{ij}\otimes\frac{1}{2}(I + \mathcal{U}_{i,j})
\end{equation}
where $\mathcal{U}_{i,j} = U_{D''}U_{\tilde{1}{b_j}}U_{D'}U_{\tilde{1}_{b_i}}U_{D}$. Measuring the ancilla registers, they output the bitstrings $i, j$ with probability
\begin{equation}
    P(i,j) = \frac{1}{2|B|^2}\left(1 + \text{Re}\left(\bra{\psi_0}\mathcal{U}_{i,j}\ket{\psi_0}\right)\right)
\end{equation}
which provides data on the strength of the response in the $2$-d region $b_i \times b_j$. 

To implement this, we need access to $U_H$ to form $\{\mathcal{U}_{i,j}\}_{i,j=1}^{|B|}$ and its controlled version. Each one of the $\mathcal{U}_{ij}$ can be implemented with a degree $d = O(|B|)$ polynomial. Since there are $O(|B|^2)$ of these, each one requiring $O(|B|)$ queries to $U_H$ gives a total query complexity to $U_H$ that is $\widetilde{O}(|B|^3)$. Using the same procedure as above, we sample from a distribution over $O(|B|^2)$ elements and determine a pair $(i,j)$ corresponding to the region $b^{0}_i \times b^{0}_j$ which has the largest response with some probability. We then subdivide the box $b^{0}_i\times b^{0}_j$ into into $|B|$ bins, each of size $O(1/|B|^4)$. This will require $\widetilde{O}(|B|^5)$ queries to $U_H$ to implement and $\widetilde{O}(1)$ repetitions to determine a new bin $b^{1}_i\times b^{1}_j$. We wish to repeat this $k$ times until $|b^{k}_i \times m^{k}_j| = O(\gamma^2/\alpha^2)$. This means that $k = O(\log(\gamma/\alpha))$. Then, the total number of queries to $U_H$ is 
\begin{equation}
    \sum_{i=1}^{k}|B|^{2i+1} = O(|B|^{k^2}) = \widetilde{O}\left(\left(\frac{\alpha}{\gamma}\right)^2\right).
\end{equation}
But the success probability for this procedure requires an additional number of queries that is $O(\alpha^2 \beta^3)$, where we assume that the subnormalization factor $\beta$ is the same for all the dipole operators. Then, the total query complexity is 
\begin{equation}
    \widetilde{O}\left(\frac{\alpha^4 \beta^3}{\gamma^2}\right).
\end{equation}

This binary-search inspired procedure for finding frequency bins of pronounced response can be generalized to any order. In the case of $n$th order response functions, where one desires to determine a $n$-dimensional bin of frequencies $b_0 \times b_1 \times \cdots \times b_{n-1}$ where the size of the bin is order $O\left(\left(\gamma/\alpha\right)^n\right)$, will in general require $\widetilde{O}\left((\alpha/\gamma)^n\right)$ queries to $U_H$ to form the desired polynomial. An additional factor of $O(\alpha^{n}\beta^{n+1})$ arises due to the success probability of implementing the $\mathcal{U}$'s, so that the total complexity is expressed as
\begin{equation}
\widetilde{O}\left(\frac{\alpha^{2n}\beta^{n+1}}{\gamma^n}\right)    
\end{equation}
queries to $U_H$. These can then be used to evaluate the $I_{ij}$ portions of the terms in Eq. \eqref{eq:thirdord}.

Once the region has been determined where the response is desired accurately, one can either Monte-Carlo sample or use amplitude estimation which increases the query complexity by a factor of $O(1/\epsilon^2)$ and $O(1/\epsilon)$ respectively. The amplitude estimation based sampling algorithm will therefore scale as
\begin{equation}
    \widetilde{O}\left(\frac{\alpha^{2n}\beta^{n+1}}{\gamma^{n}\epsilon}\right)
\end{equation}
queries to $U_H$ to measure the response in some small frequency bin to precision $\epsilon$. Of course, this will need to be repeated for many bins, so the cost is additive in total number of bins one needs to consider to accurately compute the terms corresponding to the $d^k_{lm}$'s in Eq. \eqref{eq:thirdord}. Therefore, this methodology could in practice be used to compute high-precision estimates to nonlinear quantities in optical and molecular physics.

Using our estimates for $\alpha$ and $\beta$ used in block-encoding $H_0$ and the dipole operators in second quantization gives an expected query complexity to the molecular Hamiltonian block encoding for $n$th order molecular response quantities that is
\begin{equation}
    \widetilde{O}\left(\frac{N^{5n+1}\eta^{n+1}}{\gamma^n\epsilon}\right).
\end{equation}
The bulk of the complexity of this algorithm comes from two sources, the success probability of the block-encodings and the spectral resolution required. Note that in many cases the success probability estimates can be very pessimistic, as the numerator can be quite large and compensate for a large subnormalization factor in the denominator. However, improving over the $\gamma^{n}\epsilon^{-1}$ resolution and precision scaling would imply a violation of the Heisenberg limit, so we expect that the factor of $O(\gamma^n\epsilon^{-1})$ cannot be improved in general.

\section{\label{sec:level4}Discussion\protect }
\subsection{Comparison with previous work}
In this work, we have shown how one can calculate  molecular response properties for the fully correlated molecular electronic Hamiltonian using the dipole approximation for $n$th order response functions in a manner that can be implemented in $O\left(\alpha^{2n}\beta^{n+1}/(\gamma^{n}\epsilon)\right)$ queries to the Hamiltonian block encoding, with $\alpha$ and $\beta$ the block encoding subnormalization factors for the Hamiltonian and the dipole operators respectively. In this section we will use the estimates  $\alpha = O(N^2)$ and $\beta = O(N\eta)$. Due to electronic correlation, many classical methods such as time-dependent density functional theory can fail to capture correct quantitative and even qualitative behavior of the absorption strengths and excitation energies of molecules interacting with light. Since quantum algorithms can efficiently time-evolve the full molecular electronic Hamiltonian, quantum algorithms can efficiently capture highly correlated electronic motion. Accordingly, there has been interest before our work in developing quantum algorithms for estimating linear response functions. There are multiple works that we are aware of that report algorithms \cite{caiQuantumComputationMolecular2020, huangVariationalQuantumComputation2022,kosugiLinearresponseFunctionsMolecules2020,roggeroLinearResponseQuantum2019} for computing similar spectroscopic quantities to those we compute in the above sections.

In the work of Ref. \cite{caiQuantumComputationMolecular2020}, the response function is calculated by solving a related linear system with a quantum linear systems algorithms such as HHL \cite{harrowQuantumAlgorithmSolving2009}. They report a $\poly(N)$ complexity but do not discuss the degree of this polynomial. Their approach is to solve the linear system
\begin{equation}
\begin{aligned}
    A(\pm \omega) &:= H_0 - E_0 \mp (\omega + i \gamma)\\
    A(\omega)\ket{x} &= D\ket{\psi_0},
\end{aligned}
\end{equation}
by block encoding the dipole operator $D$. Solving this system for a sequence of $\omega$'s with spacing smaller than the damping parameter $\gamma$ and for each direction, one can obtain an approximation to the frequency dependent polarizability. One difficulty of this approach is that HHL requires many additional ancilla qubits to perform the classical arithmetic to compute the $\arccos$ function of the stored inverted eigenphases as well as a dependence on the condition number of the linear system. The scaling of HHL is $\widetilde{O}(\kappa^2/\epsilon)$, where $\kappa$ is the condition number of $A(\omega)$. However, this cost in quantum linear solvers can be improved by using more recent work based on quantum signal processing (QSP) and the discrete adiabatic theorem. These methods can solve the linear system problem with $\widetilde{O}(\kappa \log(1/\epsilon))$ \cite{costaOptimalScalingQuantum2021} queries to a block encoding of $A(\omega)$. Once the quantum state corresponding to a solution is prepared, measurements of the observable $\bra{\psi_0}D' D \ket{\psi_0}$ as well as $\bra{x}A^\dagger(\omega)A(\omega)\ket{x}$ are required to extract similar information to that in our algorithm. Since the number of terms in $A(\omega)$ is the same as the number of terms in $H$, $O(N^2)$, the number of terms is $A^\dagger A = O(N^4)$. This will in general require $O(N^4/\epsilon^2)$ repetitions of the circuit to obtain an estimate to $\bra{x}A^\dagger(\omega)A(\omega)\ket{x}$, assuming the variance is $O(1)$ \cite{McClean2016a}. We take as  estimates that $\kappa(A) = \widetilde{O}(||H||/\gamma) = \widetilde{O}(N^2/\gamma)$, and that the success probability for block-encoding for the dipole with subnormalization $\beta = O(N\eta)$ is improved with amplitude amplification. Using the HHL-based algorithm originally presented in Ref.~\cite{caiQuantumComputationMolecular2020}, the expected scaling is $\widetilde{O}(N\eta \kappa^2 N^4/\gamma^2\epsilon^3)$ which simplifies to $\widetilde{O}(\eta N^9/\gamma^2\epsilon^3)$ queries to the time evolution operator for the linear operator $A(\omega)$. Using more modern techniques based on block encoding, we obtain the scaling $\widetilde{O}(N^7\eta/\gamma\epsilon^2)$ queries to the block encoding of $U_H$ for estimating the response at a particular frequency using the algorithm in Ref.~\cite{caiQuantumComputationMolecular2020} to $\epsilon$ accuracy.

Another algorithm reported in Ref. \cite{kosugiLinearresponseFunctionsMolecules2020} presents a quantum algorithm for computing response quantities such as the polarizability in linear response theory. Similar to our work, the authors begin by assuming the observable $\mathcal{A}$ they wish to compute the response function for has a one-body decomposition in the electron spin-orbital basis
\begin{equation}
    \mathcal{A} = \sum_{m,m'}\mathcal{A}_{m,m'}a_m^\dagger a_m',
\end{equation}
and access to the ground state $\ket{\psi_0}$. Here, instead of directly block-encoding $\mathcal{A}$, the authors block-encode individual matrix elements of $\mathcal{A}$ and apply it to the ground state. For a particular matrix element $\mathcal{A}_{mn}a_m^\dagger a_n$, they decompose into an LCU using the JW or Bravyi-Kitaev transformation, with off-diagonal terms having a representation as the linear combination of $8$ unitary matrices. This block encoding will have a success probability that depends on the overlap, $\bra{\psi_0} a^\dagger_m a_n \ket{\psi_0}$, which will be small in regions where the $n \rightarrow m$ response is also small. Indeed, this block-encoding of a projector can be exponentially small or even zero if the particular transition is restricted by the problem symmetry. If the block-encoding is successful, then the state $a_m^\dagger a_n \ket{\psi_0}$ will be prepared in the system register. Then the authors apply QPE on the state $a_m^\dagger a_n \ket{\psi_0}$ and repeatedly measure until the phase estimation procedure successfully returns the bitstring corresponding to the energy of the associated transition. This work additionally numerically simulates their algorithm for the $C_2$ and $N_2$ molecules, and finds reasonable fitting to the response functions computed classically with full-configuration interaction. Although explicit circuits are presented which implement this algorithm, the computational complexity of this algorithm is not discussed in detail.

Additionally, a variational algorithm is proposed in Ref. \cite{huangVariationalQuantumComputation2022} to approximate the first-order response of molecules to an applied electric field. This algorithm operates similarly to Ref. \cite{caiQuantumComputationMolecular2020}, but instead uses a variational ansatz on the solution state $\ket{\psi(\boldsymbol{\theta})} \sim \ket{x}$  to minimize a cost function resulting from a variational principle for the polarizability. 
This computation also has a similar overhead in computing expectation values as Ref.~\cite{huangVariationalQuantumComputation2022}, but without guarantees of accuracy and circuit depth that can be provided by fully coherent quantum algorithms such as HHL, QPE, and QSP.

\subsection{Conclusions and Future Work}
In this work we have presented a quantum algorithm for computing general linear and non-linear absorption and emission processes which can be used in conjunction with simple classical post-processing to approximate response functions in molecules subjected to an external electric field treated semi-classically. Our work extends previous research which have approached this problem very differently. We compare our approach to these methods, and in cases where rigorous complexity estimates can be obtained, show that our approach scales more favorably than previous methods. Our algorithm approaches the difficult problem of obtaining high-fidelity estimates to the desired matrix elements with a quantum approach, and combines these results on a classical computer to obtain a parameterized representation of the response function for the given molecule. Furthermore, our algorithm is conceptually simple and employs a measurement scheme that reduces the computational effort over regions where the response is negligible. We showed that our quantum subroutine combined with the feedback from measurements is nearly optimal at determining excitation energies of the Hamiltonian. In addition, beyond the ancilla qubits for block-encoding the desired operation, our algorithm can be effectively reduced to a Hadamard test for non-unitary matrices which requires only a single additional ancilla; or a linear combination of such circuits. 

We additionally presented a method to find the dominant contributions to the response function with a systematically improvable procedure to resolve spectral regions where the response is most pronounced at the Heisenberg limit. Once the region of interest has been determined to the desired resolution, then one can either sample at the standard classical limit or use amplitude estimation procedure to estimate the response function magnitude in a particular spectral region. Finally, our algorithm can be used to compute response functions to any desired order through iterative applications of the algorithm presented for linear response, so long as the electric dipole approximation accurately describes the desired physics. Of course as one goes to higher order of response function, there is a combinatorial explosion in the number of terms that need to be computed that is unavoidable for any current algorithm. However, it is often the case that response functions are only calculated or experimentally determined up to some small constant order or a very small number of fixed terms.

Additionally, the subnormalization factors also grow exponentially in the desired order of response. Near the end of this work, we became aware of a work published recently \cite{dingQuantumMultipleEigenvalue2024}, which performs a similar task with a different approach. In their Heisenberg-limited scaling algorithm, they perform a Gaussian fitting based on Hadamard tests of the Hamiltonian simulation matrix performed for times drawn from a probability density formed by the Fourier coefficients. In both cases, it will be of interest to study the robustness of these methods to imperfections in the ground state preparation.

In many instances the methods for computing these response functions do not permit provable guarantees on the accuracy of the result. However, using the power of quantum computing, it is tractable to perform eigenvalue transformations on very high dimensional systems making this algorithm efficient. Future work investigating other paradigms beyond block encoding that do not have possible issues with success probabilities could be of interest. In lieu of this, it may be possible that additional analysis can be performed to improve the estimates on the success probability for block encoding the dipole operators, using spectroscopic principles such as the Thomas-Reiche-Kuhn sum rule. Furthermore, since our algorithm only includes electronic degrees of freedom, it is a topic of future work to study how the asymptotic scaling of this algorithm compares when including vibronic degrees of freedom.

\begin{acknowledgments}
{T.D.K was supported by a Siemens FutureMakers Fellowship and by the U.S. Department of Energy, Office of Science, Office of Advanced Scientific Computing Research under Award Number DE-SC0023273.} 
T.F.S. {was supported}as a Quantum
Postdoctoral Fellow at the Simons Institute for the Theory of Computing, supported by the U.S. Department of
Energy, Office of Science, National Quantum Information
Science Research Centers, Quantum Systems Accelerator. T.F.S. also acknowledges partial support from the 2019 Microsoft Research internship program. 
{L.K. was supported by the US Department of Energy, Office of Science, Office of Workforce Development for Teachers and Scientists, Office of Science Graduate Student Research (SCGSR) program. The SCGSR program is administered by the Oak Ridge Institute for Science and Education for the DOE under contract number DE-SC0014664.}
 {K. B. W. was supported by the National Science Foundation (NSF) Quantum Leap Challenge Institutes (QLCI) program through Grant No. OMA-2016245.}
{TDK also wishes to acknowledge Leonardo Coello-Escalante at the University of California, Berkeley and Szabolcs Goger at the Univeristy of Luxembourg for insightful
discussions on classical computational spectroscopy methods, and thank Zhiyan Ding at the University of California, Berkeley for helpful discussions on the eigenstate filtering portion of this work as well as} for reviewing the manuscript prior to submission. We also thank the Institute for Pure and Applied Mathematics (IPAM) for the semester-long program “Mathematical and Computational Challenges in Quantum Computing” in Fall 2023 during which this work was discussed.
\end{acknowledgments}

\bibliography{ExcitedStates.bib}

\newpage
\onecolumngrid

\appendix
\section{\label{app: Filtering} Filtering}
Suppose we have access to an $(\alpha, m, 0)$ block encoding of the second quantized molecular Hamiltonian $U_H$. We assume that the initial state $\ket{\psi}$ has been prepared with overlap $p_0 = \left\vert \braket{\psi_0}{\psi}\right\vert^2$ and that the spectral gap $\Delta = \lambda_1-\lambda_0$, neither of which are exponentially small. With these assumptions, the eigenstate filtering method of \cite{tongQuantumEigenstateFiltering2022,linOptimalPolynomialBased2020a}, promises an efficient algorithm which can prepare the ground state with
\begin{equation}
    O\left(\frac{\alpha}{\sqrt{p_0} \Delta}\log\left(\frac{1}{\epsilon}\right)\right)
\end{equation}
queries to $U_H$, where $||\ket{\psi_0}-\ket{{\psi}}||^2 \leq \epsilon$. Given these constraints, we can apply this algorithm to efficiently transform $\ket{\psi}$ into a high fidelity approximation of the true ground state $\ket{\psi_0}$. Once the high fidelity ground state is in hand, we can additionally use eigenstate filtering to obtain a high-fidelity estimate to the ground state energy as well. Although the assumption of having an exact ground state is unlikely to be achieved in practice, this algorithm allows us to transform our initial state with $1/\text{poly}(n)$ overlap, to a state with arbitrarily close-to-1 overlap with the true ground state. This step is considered to be a preprocessing step, as it allows us to make the initial assumption of exact ground state preparation. 

Next, we define the scalar indicator function as
\begin{equation}
    1_{[a,b]}(x) :=\begin{cases}
        1 & x \in [a,b]\\
        0 & \text{ else}.
    \end{cases}
\end{equation} 
 To motivate this construction, suppose that we are able to exactly implement the indicator function of our Hamiltonian through its $(\alpha, m, 0)$ block encoding and QSP. We additionally assume access to $(\beta, k, 0)$ and $(\beta', k', 0)$ block encodings of two dipole operators $D$ and $D'$ respectively.  Given access to these three operations, we define the map $\mathcal{U}$ according to its action on the ground state and ancilla registers as
 \begin{equation}
     \mathcal{U}\ket{\psi_0}\ket{0}_{m}\ket{0}_{k}\ket{0}_{k'} = D' 1_{[a,b]}(H/\alpha\beta\beta') D\ket{\psi_0}\ket{0}_m\ket{0}_k\ket{0}_{k'} + \ket{\perp}.
 \end{equation}
 Now, we apply $\mathcal{U}$ in a controlled manner using the Hadamard test circuit shown in Fig.~\ref{fig:Hadamard Test appdx},
 \begin{figure}[ht]
    \centering
    \begin{tikzpicture}
    \node[scale = 1.5]
    {
        \begin{quantikz}
        \lstick{$\ket{0}$}&\qw&\gate{H}&\ctrl{0}\vqw{1}&\gate{H}&\meter{}\\
        \lstick{$\ket{\psi_0}$}&\qw&\qw&\gate[2]{\mathcal{U}}&\qw&\qw\rstick{$ \frac{D' 1_{[\omega \pm \Delta]}(H) D}{(\alpha+\omega)\beta\beta'} \ket{\psi_0}$}\\
        \lstick{$\ket{0}_{m+k+k'}$}&\qw&\qw& &\qw&\qw\rstick{$\bra{0}_{m+k+k'}$}
        \end{quantikz}
    };
    \end{tikzpicture}
    \caption{A circuit diagram description of the Hadamard test for non-unitary matrices that upon measuring the ancilla qubit provides an estimate of $\text{Re}\langle \psi_0 | \mathcal{U} | \psi_0 \rangle$}
    \label{fig:Hadamard Test appdx}
\end{figure}
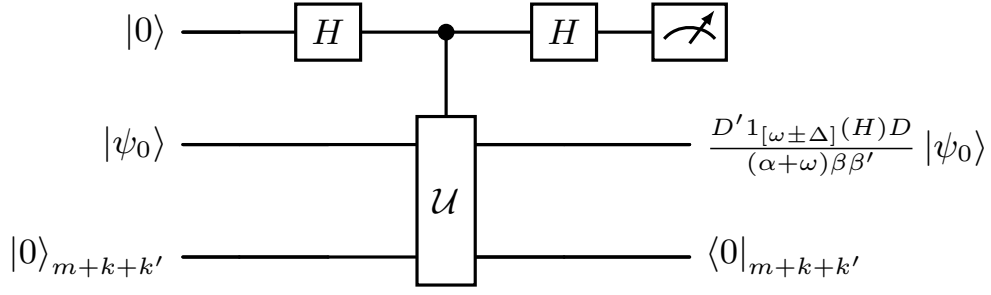
which (up to normalizing constants) performs the following transformations
\begin{align}&\ket{\psi_0}\ket{0}_{m}\ket{0}_{k}\ket{0}_{k'}\ket{0} + \ket{\psi_0}\ket{0}_{m}\ket{0}_{k}\ket{0}_{k'}\ket{1}\notag\\
    &\rightarrow \ket{\psi_0}\ket{0}_{m}\ket{0}_{k}\ket{0}_{k'}\ket{0} + \left(\frac{1}{\alpha\beta\beta'} D' 1_{[a,b]}(H) D\ket{\psi_0}\ket{0}_m\ket{0}_k\ket{0}_{k'} + \ket{\perp}\right)\ket{1}\notag\\
    &\rightarrow \left(\ket{\psi_0}\ket{0}_{m}\ket{0}_{k}\ket{0}_{k'} + \frac{1}{\alpha\beta\beta'}D' 1_{[a,b]}(H) D\ket{\psi_0}\ket{0}_m\ket{0}_k\ket{0}_{k'} + \ket{\perp}\right)\ket{0} +\ket{\perp}.
\end{align}
This will give a success probability of measuring the ancilla qubit in the $\ket{0}$ state as
\begin{equation}
    P(0) = \frac{1}{2}\left(1+ \frac{1}{\alpha\beta\beta'}\bra{\psi_0}D' 1_{[a,b]}(H) D\ket{\psi_0}\right).
\end{equation}

Now, we will expand and simplify $\bra{\psi_0}D' 1_{[a,b]}(H) D\ket{\psi_0}$ as the following
\begin{align*}
    \bra{\psi_0}D' 1_{[a,b]}(H) D\ket{\psi_0}
    &=\bra{\psi_0}D' 1_{[a,b]}(H) \sum_{n>0}d_{0,n}\ket{\psi_n}\\
    &=\bra{\psi_0}D'\sum_{n>0}d_{0,n}\sum_{\lambda_j \in [a,b]}\ket{\psi_j}\!\bra{\psi_j}\psi_n \rangle\\
    &=\bra{\psi_0}D'\sum_{n>0}d_{0,n}\sum_{\lambda_j \in [a,b]}\ket{\psi_j}\delta_{j,n}\\
    &=\bra{\psi_0}\sum_{\lambda_j \in [a,b], j\neq 0}d_{0,j}D'\ket{\psi_j}\\
    &=\bra{\psi_0}\sum_{\lambda_j \in [a,b], j\neq 0}\sum_{n \neq j}d_{0,j}d'_{j,n}\ket{\psi_n}\\
    &=\sum_{\lambda_j \in [a,b], j\neq 0}\sum_{n \neq j}d_{0,j}d'_{j,n}\braket{\psi_0}{\psi_n}\\
    &=\sum_{\lambda_j \in [a,b], j\neq 0}\sum_{n \neq j}d_{0,j}d'_{j,n}\delta_{0,n}\\
    &=\sum_{\lambda_j \in [a,b], j\neq 0}d_{0,j}d'_{j,0}.
\end{align*}
Note that if $D' = D$, then we recover the case for the response function corresponding to a diagonal term of the polarizability tensor.

\subsection{Application to higher-order polarizabilities}
Let us now consider how this algorithm works for higher-order polarizibilities. Consider a set of windows $W, W', W''$ and a set of directions for the dipole operators, $D^x, D^y, D^z$. We wish to apply the sequence 
\begin{align*}
    \mathcal{U}_2 = D^z 1_{W'}(H)D^y1_{W}(H)D^x 
\end{align*}
in a controlled manner, similar to what was done above. We know that as above, this corresponds to an LCU of $\mathcal{U}_2$ and $I$ applied to the ground state $\ket{\psi_0}$. The success probability of the procedure is related to the desired inner product in an analagous manner to the standard Hadamard test,
\begin{equation}
    P_{\text{succ}} = \frac{1}{2}\left(1 + \Re{\bra{\psi_0}\mathcal{U}_2\ket{\psi_0}}\right).
\end{equation}

The important term is $\bra{\psi_0}\mathcal{U}_2\ket{\psi_0}$ which we evaluate below.
\begin{align}
    \bra{\psi_0}\mathcal{U}_2\ket{\psi_0} &= \bra{\psi_0}D^z1_{W'}(H)D^y1_{W}(H)D^x \ket{\psi_0}\notag\\
    &=\bra{\psi_0}D^z1_{W'}(H)D^y1_{W}(H)\sum_{n>0}d^{x}_{0,n}\ket{\psi_n}\notag\\
    &=\bra{\psi_0}D^z1_{W'}(H)D^y \sum_{\lambda_n \in W}d^{x}_{0,n}\ket{\psi_n}\notag\\
    &=\bra{\psi_0}D^z1_{W'}(H) \sum_{\lambda_n \in W}d^{x}_{0,n}\sum_{k\neq n}d^{y}_{n,k}\ket{\psi_k}\notag\\
    &=\bra{\psi_0}D^z\sum_{\lambda_k \in W'}\sum_{\lambda_n \in W}d^{x}_{0,n}d^{y}_{n,k} \ket{\psi_k}\notag\\ 
    &=\bra{\psi_0}\sum_{l\neq k}\sum_{\lambda_k \in W'}\sum_{\lambda_n \in W}d^{x}_{0,n}d^{y}_{n,k} d^{z}_{k,l}\ket{\psi_l}\notag\\ 
    &=\sum_{\lambda_k \in W'}\sum_{\lambda_n \in W}d^{x}_{0,n}d^{y}_{n,k} d^{z}_{k,0}
\end{align}
which is exactly the kind of terms we need for the first hyperpolarizability. 

It is straightforward to show that repeatedly applying filters and dipole operators in the manner as above that this procedure can be used to compute $n$th-order moments of the dipole operator and extract frequency-dependent quantities for arbitrary regions of the $n$-dimensional excitation energy space.

\subsection{Bounding the error in approximate indicator function}
We will need to obtain bounds to a polynomial approximation of a (smoothed) indicator function through a mollified indicator function $\tilde{1}_{[a,b]}(x)$. Then, once we obtain bounds for the degree of polynomial needed to approximate the indicator to the desired precision, we can apply the machinery of quantum signal processing to implement the matrix function $\tilde{1}_{[a,b]}(H)$. There are two points where the jump discontinuity occurs, namely at $x = a$ and $x = b$, we will consider the regions $[a-\delta, a+\delta]$ and $[b-\delta, b+\delta]$ around those points. We will follow the construction given in \cite{lowHamiltonianSimulationUniform2017}, where the sign and indicator functions are approximated by $\text{erf}(k x)$, for $k > 1$. 

First, we need to construct a polynomial approximation to the function $\text{erf}(kx)$. We can cite Corollary 4 from Appendix A of \cite{lowHamiltonianSimulationUniform2017}.

\begin{theorem}[Polynomial Approximation to the error function $\text{erf}(kx)$ (Corollary 4 \cite{lowHamiltonianSimulationUniform2017}]
    For every $k > 0$, $0<\epsilon \leq 1$, the odd polynomial $p_{\erf,k,n}$ of odd degree $n = \mathcal{O}\left(\sqrt{\left(k^2 + \log\left(1/\epsilon\right)\right)\log(1/\epsilon)}\right)$ satisfies
    \begin{equation}
    \begin{aligned}
        p_{\erf,k,n} &= \frac{2k e^{-k^2/2}}{\sqrt{\pi}}\left(I_0(k^2/2)x + \sum_{j=1}^{(n-1)/2}I_j(k^2/2)(-1)^j \left(\frac{T_{2j+1}(x)}{2j+1}-\frac{T_{2j-1}(x)}{2j-1}\right)\right),\\
        \epsilon_{\erf,k,n} &= \max_{x\in[-1,1]|}|p_{\erf,k,n}(x) - \erf(kx)|\leq \frac{4k}{\sqrt{\pi}n}\epsilon_{\text{gauss},k,n-1}\leq \epsilon
    \end{aligned}
    \end{equation}
    Where $I_j(x)$ is the modified Bessel function of the first kind and $\epsilon_{\text{gauss},k,n}$ is the maximum error of a degree $n$ polynomial approximation to the Gaussian $e^{-(k x)^2}$ and $T_j$ is the $j$th Chebyshev polynomial of the first kind.
    \label{thm: polynomial approximation to erf(x)}
\end{theorem}

We also need a result for the rate at which the error function $\erf(kx)$ converges to the step function as we increase $k$. This is provided by Lemma 10 of Appendix A of \cite{lowHamiltonianSimulationUniform2017}.
\begin{lemma}[Entire Approximation of the sign function $\text{sgn}(x)$ (Lemma 10 Ref. \cite{lowHamiltonianSimulationQubitization2019})]
For every $\delta > 0, x \in \mathbb{R}, \epsilon \in (0,\sqrt{2/e\pi})$, let $k = \frac{\sqrt{2}}{\delta}\log^{1/2}(2/\pi \epsilon^2)$. Then the function $f_{\text{sgn},\delta,\epsilon}(x) = \erf(kx)$ satisfies
\begin{equation}
\begin{cases}
    1 &\geq |f_{\text{sgn},\delta,\epsilon}(x)|,\\
    \epsilon &\geq \max_{|x|\geq \delta/2}|f_{\text{sgn},\delta,\epsilon}(x) - \text{sgn}(x)|,
\end{cases}
\end{equation}
    and
    \begin{equation}
        \text{sgn}(x) = \begin{cases}
            1, \hspace{.4cm} & x>0,\\
            -1, \hspace{.4cm} & x<0,\\
            \frac{1}{2}, \hspace{.4cm} & x = 0.
        \end{cases}
    \end{equation}
    \label{lem: erf convergence}
    Here $\delta$ is the parameter that controls the interval over which the transition occurs.
\end{lemma}

It is easy to show that given an approximation to the sign function, the symmetric rectangle function can be implemented as a sum of two approximate sign functions. Since by Lemma \ref{lem: erf convergence}, we have that $k = O\left(\frac{\sqrt{\log(1/\epsilon^2)}}{\delta}\right)$ and by Theorem \ref{thm: polynomial approximation to erf(x)} that the polynomial degree 
$d = O\left(\sqrt{\log(1/\epsilon)\left(k^2+\log(1/\epsilon)\right)}\right)$, we can choose a degree 
\begin{equation}
    d = O\left(\log(1/\epsilon)\delta^{-1}\right)
\end{equation}
to satisfactorily approximate the error function and in turn the indicator function.

\begin{figure}
    \centering
    \includegraphics[width=.6\paperwidth]{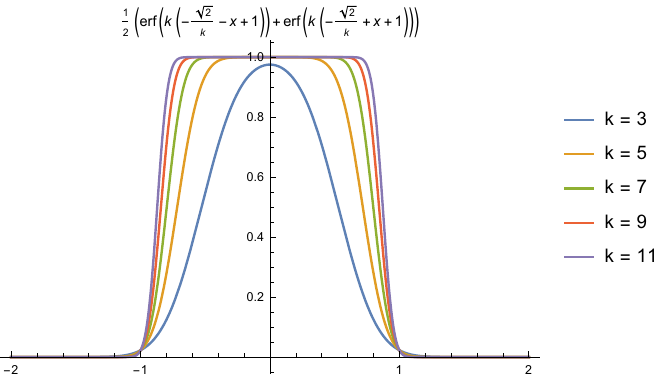}
    \caption{Error function approximation to the indicator function on $[-1,1]$ as a function for various choices of $k = \widetilde{O}(1/\delta)$.}
    \label{fig:indicator appx}
\end{figure}

Then, given $\epsilon/2$-accurate approximations to the polynomial approximation to the error function, which in turn is being used to approximate the sign function, we can construct a $\epsilon$-accurate approximation to the indicator function corresponding to some symmetric window of width $\Delta$, with a shift $\kappa = \frac{\Delta+\delta}{2}$ by
\begin{equation}
    1_{[-\Delta, \Delta]}(x)\sim \frac{1}{2}\left(\erf(k(x+\kappa))+\erf(k(\kappa -x))\right) 
\end{equation}
for sufficiently large $k$. Furthermore by Lemma 30 of \cite{gilyenQuantumSingularValue2019a}, we are guaranteed that this polynomial is even, therefore it can be constructed with only a single ancilla qubit with QSP.

We would like to study the error in the region near a jump. In addition, we would like to know how the polynomial approximation error propagates into the response we are planning to approximate. These are characterized by the following theorems.

\begin{lemma}[Bounding the error around the jump]
    The error in approximating the sign function on the interval $[-\delta, \delta]$ with the error function $\erf(kx)$ with $k = \widetilde{O}(1/\delta)$ accumulates at most linearly in the interval $[-\delta, \delta]$.
    \label{lem: Bounding error around jump}
\end{lemma}
\begin{proof}
    Consider the following,
    \begin{equation}
        \int_{-\delta}^\delta |\erf(kx) - \text{sgn}(x)|dx
    \end{equation}
    which can be expressed as 
    \begin{equation}
        \lim_{\epsilon\rightarrow 0}\int_{-\delta}^{-\epsilon} |\erf(kx) + 1|dx +\int_{\epsilon}^{\delta}|\erf(kx)-1|dx,
    \end{equation}
    for $\delta>\epsilon>0$.
    We would like to study the rate at which the error in the region grows as a function of $\delta$. Let
    \begin{equation}
        f(kx) = |\erf(kx) - 1|
    \end{equation}
    and
    \begin{equation}
        g(\delta; \epsilon) = \int_{\epsilon}^{\delta}f(kx) dx.
    \end{equation}
    However, $f(kx)$ has implicit $\delta$ dependence since we always choose our $k = \widetilde{O}(\delta^{-1})$. If we substitute this estimate into $k$, we have that 
    \begin{equation}
        f(x/\delta) := |\erf(x/\delta) - 1| \sim f(kx).
    \end{equation}

    Now, let $y = x/\delta$, so we can write
    \begin{equation}
        g(\delta;\epsilon) = \int_{\epsilon/\delta}^{1}\delta f(y) dy.
    \end{equation}
    Then, if we calculate $\frac{d}{d\delta} g(\delta;\epsilon)$ and apply the fundamental theorem of calculus, we find
    \begin{equation}
        \frac{d g(\delta;\epsilon)}{d \delta } =\frac{\epsilon}{\delta}f(\epsilon/\delta) + \int_{\epsilon/\delta}^1 f(y) dy.
        \label{eq: derivative}
    \end{equation}

    Assuming the limit and derivatives both exist, since the limit with respect to $\epsilon$ is independent of the derivative with respect to $\delta$ we have that $\epsilon\rightarrow 0$ yields
    \begin{equation}
        \frac{d g(\delta;0)}{d \delta } =  \int_{0}^1 f(y) dy.
        \label{eq: limit}
    \end{equation}
    Then, since
    \begin{equation}
        \int_{0}^1 f(y) dy = \int_{0}^1 |\erf(y) - 1| dy \simeq .51394 = O(1),
        \label{eq: integral}
    \end{equation}
    We have that 
    \begin{equation}
        \frac{d g(\delta;0)}{d \delta } = O(1) \implies g(\delta;0) = O(\delta).
    \end{equation}
    Therefore, the error grows linearly in the $\delta$-region around the jump. This result is verified via a numerical experiment shown in Fig. \ref{fig:delta err}.
    The bound for the interval from $[-\delta, -\epsilon]$ for the function $|\erf(kx) + 1|$ is similar.
\end{proof}

\begin{figure}[h]
    \centering
    \includegraphics[width=.6\paperwidth]{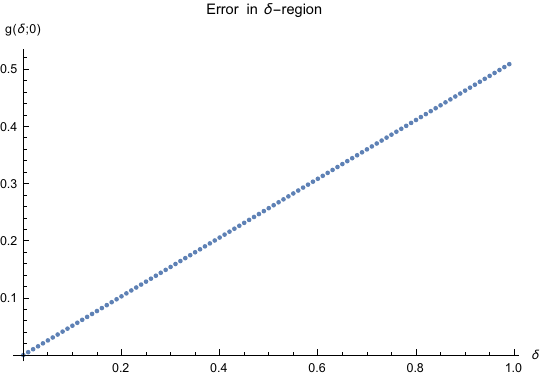}
    \caption{$g(\delta; \epsilon = .0001)$, which measures error accumulation of the error function approximation to the sign function over the smoothing interval $[.0001, \delta]$ with $\delta \in [.0002, 1]$. Notice the linear growth in the $L_2$ error in the $\delta-$interval around the jump as $\delta$ increases, verifying the estimates of Lem.~\ref{lem: Bounding error around jump} and extending the estimates of Lem.~\ref{lem: erf convergence} to the $\delta-$region around the jump.}
    \label{fig:delta err}
\end{figure}

With these results, we can now bound the error in the response function within some symmetric window $W = [w- \Delta/2, w+\Delta/2]$ after applying the dipole operators.

\begin{corrollary}[Bounding the error in estimating localized response properties using the approximate indicator function] 
Let $1_{W}(x)$ be the indicator function on the interval $W = [w-\Delta/2, w+\Delta/2]$ and $\widetilde{1}_{W}$ be the approximate indicator function on the interval $W$ as generated by the polynomial approximations to $\kappa = \left(\Delta+\delta\right)/2$ $\erf(k(x+\kappa))$ and $\erf(k(-x+\kappa))$ provided by theorem \ref{thm: polynomial approximation to erf(x)}, choosing $k$ appropriately as per lemma \ref{lem: erf convergence}. Additionally, let $D$ and $D'$ be dipole operators which have the subnormalization factors $\beta$ and $\beta'$ respectively, with $\beta' = \Theta(\beta)$. We assume that in each region there are only $\widetilde{O}(1)$ many non-zero (or non-exponentially small) amplitudes . Then, for any $\epsilon > 0$ there exists $d, \delta, N_s$ where $d$ is the degree of polynomial approximation, $\delta = O(k^{-1}\log(1/\epsilon^2))$ is the region around around the point where the jump occurs that is smoothly approximated, and $N_s$ is the number of samples needed so that we can approximate 
\begin{equation}
    \mu_{w} :=  \left\vert\bra{\psi_0}D'1_{w}(H)D\ket{\psi_0}\right\vert
\end{equation}
to accuracy $\epsilon$ in the sense that for some approximation $\hat{\mu}_{w}$ that
\begin{equation}
    |\hat{\mu}_{w} - \mu_{w}| < \epsilon.
\end{equation}

Furthermore, for a window $|w| = O(\gamma)$ there exists a quantum algorithm that makes
\begin{equation}
    O\left(\frac{\alpha^2\beta^2\log(\beta^2/\epsilon)}{\gamma \epsilon}\right)
\end{equation}
queries to the block encoding of the unperturbed Hamiltonian $H_0$ with spectrum contained in the bounded interval $[-\alpha, \alpha]$.

\label{cor:cor bounding response fn error}
\end{corrollary}

\begin{proof}
Without loss of generality we may choose $w=0$ so as to consider the symmetric region about the origin, $W=[-\Delta, \Delta]$. 

We need to consider the 3 regions where the approximation occurs. The first region is $[-2, -\delta-\Delta/2]$ (\rom{1}), the second region is the set of points within $\delta-$ball of the jump $(-\Delta/2-\delta, -\Delta/2 + \delta)$  (\rom{2}) and finally the interval between jumps $[-\Delta/2 + \delta, 0]$ (\rom{3}).  Due to the symmetry of the problem, bounding the error on these intervals is sufficient to bound the error over the entire region.

We can get a bound to the regions (\rom{1}) and (\rom{3}) using the inequality from Lemma \ref{lem: erf convergence}, namely for any $\epsilon' > 0$
\begin{equation}
     \epsilon' \geq \max_{|x\pm\Delta/2|\geq \delta}|f_{\text{sgn},\delta,\epsilon'}(x) - \text{sgn}(x)|,
\end{equation}
so immediately we have that the polynomial approximation on regions (\rom{1}) and (\rom{3}) are bounded in error by $\epsilon'$. Finally, we can bound the error in region $II$ using the result in Lemma \ref{lem: Bounding error around jump} to $O(\delta)$.

We will assume that $D$ and $D'$ can be block encoded exactly, so we will only consider the error introduced by the polynomial approximation in this proof. Therefore, as an initial step, we wish to bound
    \begin{equation}
        \left\vert\bra{\psi_0}(D'1_{W}(H)D - D'\widetilde{1}_{W}(H)D)\ket{\psi_0}\right\vert,
    \end{equation}
    which is the difference in expected values of the exact and approximate indicator functions. 

    Evaluating the above we find, 
    \begin{align}                   
    &\left\vert\bra{\psi_0}D'\sum_{i>0}d_{0,i}\left(\sum_{\omega_j\in W}\ket{\psi_j}\bra{\psi_j}- \sum_{j}\tilde{1}_{W}(\omega_j)\ket{\psi_j}\bra{\psi_j}\right)\ket{\psi_i}\right\vert\\
        &=\left\vert\bra{\psi_0}D'\left(\sum_{\omega_j\in W}d_{0,j}\ket{\psi_j}- \sum_{j}d_{0,j}\tilde{1}_{W}(\omega_j)\ket{\psi_j}\right)\right\vert \notag\\
        &=  \left\vert\bra{\psi_0}\left(\sum_{\omega_j\in W}d_{0,j}D'\ket{\psi_j}- \sum_{j}d_{0,j}\tilde{1}_{W}(\omega_j)D'\ket{\psi_j}\right)\right\vert\notag\\
        &=  \left\vert\bra{\psi_0}\sum_{\omega_j\in W}d_{0,j}\sum_{k \neq j}d'_{j,k}\ket{\psi_k}- \sum_{k\neq j}\sum_{j}d_{0,j}d'_{j,k}\tilde{1}_{W}(\omega_j)\ket{\psi_k}\right\vert \notag\\
        &=  \left\vert\sum_{\omega_j\in W}d_{0,j}d'_{j,0}- \sum_{j}d_{0,j}d'_{j,0}\tilde{1}_{W}(\omega_j)\right\vert. 
    \end{align}
    Breaking this up into sums over the various regions
    \begin{align*}
         &=  \left\vert\sum_{\omega_j\in W}d_{0,j} d'_{j,0}- \left(\sum_{|\omega_j| \in I}d_{0,j} d'_{j,0}\tilde{1}_{W}(\omega_j)+\sum_{|\omega_j| \in II}d_{0,j}d'_{j,0}\tilde{1}_{W}(\omega_j) + \sum_{|\omega_j| \in III}d_{0,j}d'_{j,0}\tilde{1}_{W}(\omega_j)\right)\right\vert\\
    \end{align*}
    which from the bounds we have above, we can write as 
    \begin{align*}
         &\leq \epsilon' + \delta + \beta\beta'\epsilon'
    \end{align*}
    where $\epsilon'$ is the accuracy to which we approximate the indicator function with $\erf(kx)$. The sum in region $\rom{3}$ is scaled by $\beta\beta'$ since we assumed that the signal in regions $\rom{1}$ and $\rom{2}$ were $O(1)$, so the terms in the rest, region $\rom{3}$, must be at contribute at most a factor of $O(\beta\beta') =O(\beta^2)$.

    Our approximation to the indicator function is through a polynomial approximation to $\erf(kx)$ where $k>\frac{\sqrt{2}}{\delta}\log^{1/2}(2/\pi \epsilon'^2)$. Using such a polynomial with $d = O\left(\log(1/\epsilon')\delta^{-1}\right)$ guarantees,
    \begin{equation}
        \max_{|x-w_0|\in [0,\Delta/2]\cup[\Delta/2 + \delta, \infty]}|\erf(kx) - 1_{w}(x)| \leq \epsilon'.
    \end{equation}

    We have three quantities to keep track of
    \begin{equation}
        \begin{aligned}
            \mu_W &:= \bra{\psi_0}D'1_W(H_0)D\ket{\psi_0}\\
            \widetilde{\mu}_W &:= \bra{\psi_0}D'\widetilde{1}_W(H_0)D\ket{\psi_0}\\
            \mu'_W &:= \frac{1}{N}\sum_{i} \widetilde{\mu}_i
        \end{aligned}
    \end{equation}
    where $\mu'_W$ is the approximation to $\widetilde{\mu}_W$ after a finite number of samples, and $\widetilde{\mu}_W$ is the exact expectation value with the polynomially approximated indicator function and $\mu_W$ is the value we wish to find.

    Given the bounds from above, we would like to find conditions on $\epsilon'', \tilde{\epsilon}$ so that whenever $|\tilde{\mu}_{W} - \mu_{W}| \leq \epsilon''$ and $|\tilde{\mu}_{W} - \hat{\mu}_{W}| \leq \tilde{\epsilon}$ $\implies |\mu_{W} - \hat{\mu}_{W}| < \epsilon$, our final desired accuracy. Observe that
    \begin{align*}
        |\mu_{W} - \hat{\mu}_{W}| &= |\mu_{W}-\tilde{\mu}_{W} + \tilde{\mu}_{W} - \hat{\mu}_{W}|\\
        &\leq \left\vert\mu_{W}-\tilde{\mu}_{W}\right\vert + \left\vert\tilde{\mu}_{W} - \hat{\mu}_{W}\right\vert\\
        &\leq \epsilon'' + \tilde{\epsilon}
    \end{align*}
    which we desire to be smaller than $\epsilon$ which implies that $\epsilon'' + \tilde{\epsilon} < \epsilon$. We are promised that $\epsilon'' \leq \epsilon' + \delta + \beta^2\epsilon'$. We desire that 
    \begin{equation}
        \epsilon'' + \tilde{\epsilon} \leq \epsilon'(1 + \beta^2) + \delta + \tilde{\epsilon} \leq \epsilon
    \end{equation}
    or
    \begin{equation}
    \frac{1}{1+\beta^2}\left( \delta + \tilde{\epsilon}\right)+\epsilon' \leq \frac{1}{1+\beta^2}\epsilon
    \end{equation}
    since we have that $\beta^2 \gg \epsilon' + \delta + \widetilde{\epsilon}$, we have that 
    \begin{equation}
        \epsilon' + O\left(\frac{1}{\beta^2}\right) \leq \frac{\epsilon}{\beta^2},
    \end{equation}
    so that $d$ must be $O(\delta \log(\beta^2/\epsilon))$, so that $\epsilon' = O\left(\frac{\tilde{\epsilon}}{{\beta^2}}\right)$. Therefore, at this stage we have
    \begin{equation}
        \delta + 2 \tilde{\epsilon} \leq \epsilon,
    \end{equation}
     or
     \begin{equation}
        \tilde{\epsilon} \leq \left\vert\frac{\epsilon -\delta}{2}\right\vert
     \end{equation}
     which is the error we have from finite sampling of the approximate state. Now we use the fact that $\delta \leq \Delta/2$, since the ramping up region must be smaller than the rectangle we are hoping to approximate. In high-precision settings, we often care about a small region satisfying $\Delta = o(\gamma)$, the desired spectral resolution, but this is rescaled by the width of the interval, which in this case is related to the spectrum of $H_0 \subset [-\alpha, \alpha]$. Since QSP requires the domain be defined on $[-1,1]$ this rescales the required $\delta \rightarrow \delta/\alpha$. Resultantly, we need a $d = O\left(\alpha/\gamma\log(\beta\beta'/\tilde{\epsilon})\right)$ polynomial to guarantee sufficient accuracy. Now that $\delta \in O(1)$ has been rescaled so that $\epsilon + \delta/\alpha = O(\epsilon)$ we have the relation that $\tilde{\epsilon} = O(\epsilon)$ which follows a standard Monte-Carlo sampling bound of 
     \begin{equation}
         N_s = O\left(\frac{1}{\epsilon^2}\right).
     \end{equation}
     Or, with the use of amplitude estimation this can be reduced to
     \begin{equation}
         N_s = O\left(\frac{1}{\epsilon}\right),
     \end{equation}
    so that the total number of queries to the block encoding of $H$ is
    \begin{equation}
        O\left(\frac{\alpha \log(\beta^2/{\epsilon})}{\gamma \epsilon}\right).
    \end{equation}
    However, the subnormalization factor from the block encoding introduces a success probability that is $O(\frac{1}{\alpha^2\beta^4})$, we use amplitude amplification with the ancilla signal state flagging the ``good subspace", in turn, this gives the query complexity to $U_H$ as 
    \begin{equation}
        O\left(\frac{\alpha^2\beta^2 \log(\beta^2/\epsilon)}{\gamma \epsilon}\right).
    \end{equation}
\end{proof}

The above result guarantees that an accurate-enough polynomial approximation to the indicator function, and enough samples of the distribution generated by this satisfactory approximation, we can determine the desired expectation value to arbitrary accuracy. In addition, it guarantees that these peaks can be resolved efficiently. We then use this algorithm with a modified binary search inspired by the procedure given in Sec 3.9 of Ref. \cite{tongQuantumEigenstateFiltering2022}. The next theorem guarantees that if the response function is ``sparsely supported", which we define in the theorem, that this algorithm can be used to find an efficient approximation to this function when the function is concentrated over a small portion of the domain. Moreover, we can determine points in the domain corresponding to large responses at the Heisenberg limit.

A result we need for this algorithm to obtain the Heisenberg limited scaling is that inequalities amongst bins can be determined to some constant probability within some fixed tolerance in an efficient manner. This can be done through an application of Hoeffding's inequality and a version of the triangle inequality for probability distributions.

\begin{lemma}[Sample size for inequality tests]
    We assume access to measurements from the probability distribution 
    \begin{equation}
        P(i) = \frac{f_{i}(x)}{\int_{-\alpha}^\alpha f(x) dx}
    \end{equation}
    where $f(x) > 0 $ and $f_{i}(x) = f(x)$ for $x \in B_i$ and $i \in \{1,\ldots, |B|\}$. We say that 
    \begin{equation}
        P(i) \gtrapprox P(j)
    \end{equation}
    if for constants $\tau > 0$ and $0 < \epsilon < 1/2$
    \begin{equation}
        \Pr(|\overline{P}_i - \overline{P}_j| > \tau) \leq \epsilon.
    \end{equation}

    Sampling from this distribution $N_s = O(\log(\epsilon^{-1})\tau^2)$ times allows us to satisfy the desired inequality to some probability $1-\epsilon$. Furthemore, with the choice $\tau = O(1/|B|)$ and $\epsilon = 1/3$, we find that $N_s = \widetilde{O}(|B|^2)$.

    If after $N_s$ samples, $|\overline{P}_i - \overline{P}_j| > \tau$ and $\overline{P}_i > \overline{P}_j$, we say that $P(i) \gtrapprox P(j)$ or vice-versa. Otherwise we say that $P(i) \approx P(j)$.

    \label{lem: inequality test}
\end{lemma}

\begin{proof}
    Let $P(i)$ be the true mean, and $\overline{P}_i = N_i/N_s$ be an approximation to $P(i)$ after $N_s$ many samples. Hoeffding's inequality provides a bound on the probability that $|P(i)$ and $\overline{P}_i$ differ by more than some fixed constant $\tau$.

    We take as the null hypothesis that $P(i) = P(j) = \mu$. We would like to bound
    \begin{equation}
        \text{Pr}(|\overline{P}_i - \overline{P}_j| \geq \tau) < \epsilon
    \end{equation}
    for some fixed constants $\tau$ and $0< \epsilon < 1/2$. 
    
    We can apply the identity,
    \begin{align*}
        &\text{Pr}(|\overline{P}_i - \overline{P}_j| \geq \tau)\\
        &=\text{Pr}(|\overline{P}_i -\mu - \overline{P}_j+\mu| \geq \tau)\\
        &\leq\text{Pr}(|\overline{P}_i -\mu|\geq \tau/2) +\text{Pr}( |\overline{P}_j-\mu| \geq \tau/2)
    \end{align*}
    
    Then we can apply Hoeffding's inequality to obtain that for any $i$, $\overline{P}_i$ satisfies
    \begin{equation}
        \text{Pr}\left(|P(i) - \overline{P}_i|\geq \tau \right) \leq e^{-N_s \tau^2} = \epsilon/2
    \end{equation}
    giving that 
    \begin{equation}
        N_s = O\left(\log(2\epsilon{-1})/\tau^2\right).
    \end{equation}
    
    Therefore, with the choice of $\tau= O(1/2|B|)$, we find
    \begin{equation}
        N_s = O\left(\log(4\epsilon^{-1})|B|^2\right).
    \end{equation}

    With $N_s$ chosen as above, we can write that 
    \begin{equation}
         \text{Pr}\left(|\overline{P}(i) - \overline{P}(j)| \geq \frac{1}{|B|}\right) \leq \epsilon
    \end{equation}
    whenever $N_s$ is $O(|B|^2 \log(1/\epsilon))$. This implies that with $N_s$ many samples a deviation of $\overline{P}_i$ from $\overline{P}_j$ of more than $\frac{1}{|B|}$ we can reject the null-hypothesis with confidence $1-\epsilon$. In this case we say that $P(i) \gtrapprox P(j)$, otherwise we say that $P(i) \approx P(j)$.
\end{proof}

\begin{figure}[h]
    \centering
    \includegraphics[]{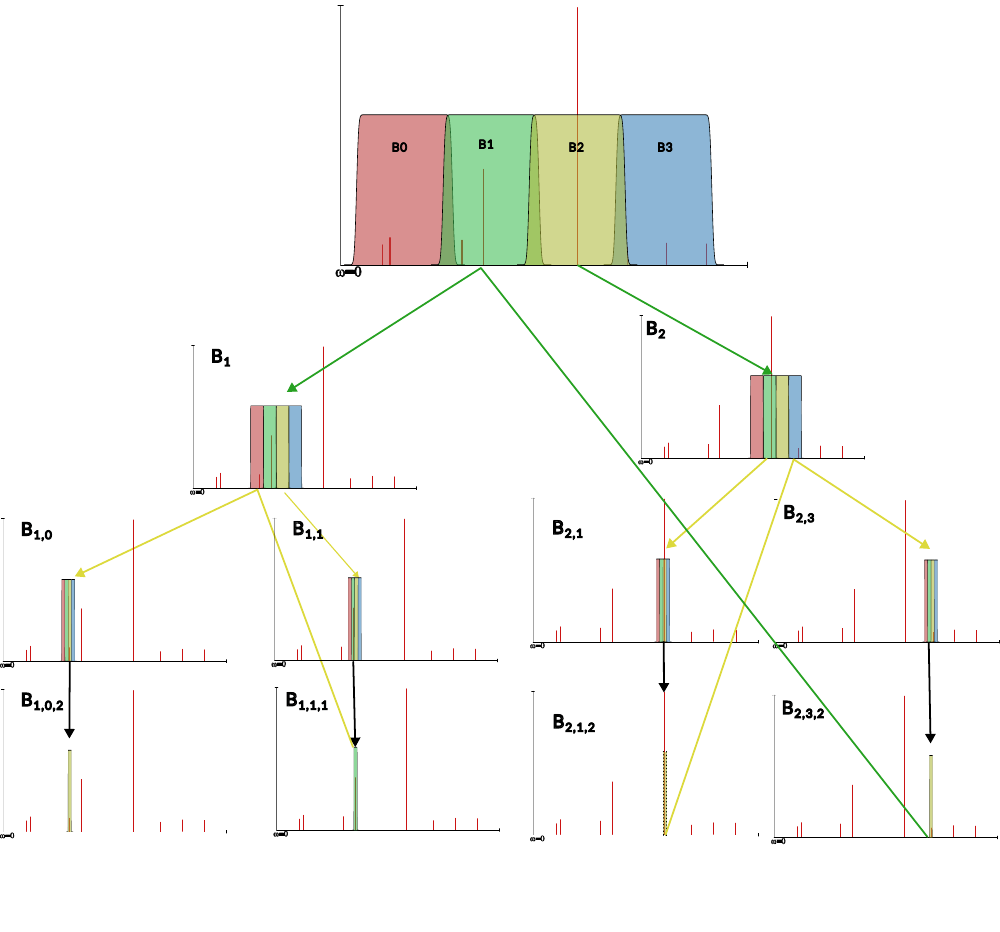}
    \caption{A diagram describing the flow of our algorithm for a few iterations in the case where there are prominent bins at each level. We sample from the distribution formed by the magnitude of response in $\{B0, B1, B2, B3\}.$ Upon sampling, we are likely to find $B2$ to be the largest, noting some response in the other bins. We then subdivide $B2$ and see that $B_{2,1}$ was most prominent. Subdivide another time, and the bin with largest response is our estimate to a excitation energy. Going back to the last time we saw a peak, we carry on to $B_{2,3}$. We then recurse to the original distribution, performing a similar task for the next most frequently sampled bin, $B1$.}
    \label{fig:schematic}
\end{figure}

\begin{theorem}[Convergence of modified binary search algorithm]
    Let $f$ be some unknown bounded function, $||f|| = O(\alpha\beta^2)$ with the property of being ``sparsely supported" i.e. that there exists some $1/2>\epsilon'>0$ and set $A \subset \text{Dom}(f) \subset [-\alpha,\alpha]$
    \begin{equation}
        f(x) = (1-\epsilon') \sum_{x \in A} f(x) + \epsilon' \sum_{x' \notin A} f(x')
    \end{equation}
    where $A$ satisfies $|A| = \widetilde{O}(1)$, suppressing polylogarithmic factors possibly depending on the dimension of the underlying approximating Hilbert space and the width of the interval $\alpha$. 
    
    By $|A| = \widetilde{O}(1)$ we mean,
    \begin{equation}
        \int_{-\alpha}^{\alpha} 1_{A}(x) dx = \widetilde{O}(1),
    \end{equation}
    is independent of $\alpha$, where 
    \begin{equation}
        1_{A}(x) =
        \begin{cases}
            1 &x\in A\\
            0 &x \notin A
        \end{cases}
    \end{equation}
    is the characteristic function for $A$.

    If we let 
    \begin{equation}
        f(\omega) = \begin{cases}
            d'_{j,0}d_{0,j} & \omega = \lambda_j - \lambda_0\\
            0 & \text{else}
        \end{cases}
    \end{equation}

    Then there exists a quantum algorithm that can approximate $\epsilon$-approximate $f(x)$, in the sense that
    \begin{equation}
        ||\widetilde{f}(x) - f(x)||_\infty \leq \epsilon
    \end{equation}
    for any $\epsilon > \epsilon'$.
    
    Moreover, this algorithm can determine some $x^* \in A$ to resolution $\gamma$ using ${O}\left(\frac{\alpha^2\beta^2\log(\beta^2/\gamma)}{\gamma}\right)$ queries to the quantum algorithm implementing $f$ in the case of linear response and ${O}\left(\frac{n\alpha^{2n}\beta^{n+1}\log(\beta^2/\gamma)}{\gamma^n}\right)$ in the case of arbitrary $n$th order response.
    \label{thm: search algorithm conv}
\end{theorem}

\begin{proof}

    For intuition, this algorithm can be viewed as a modified statistical binary search. We will use the notation
    \begin{equation}
        [|B|] = \{0, \ldots, |B|-1\}.
    \end{equation}
    
    We will assume the quantum algorithm implementing $f$ has been performed through block encoding and we assume access to those block encodings. In particular, to the block encodings of the dipole operators and block encodings of the $\widetilde{\epsilon}$ approximate indicator function applied to the (possibly shifted) unperturbed Hamiltonian $H_0$. Through the Hadamard test, we additionally assume access to
    \begin{equation}
        f(\omega) = \bra{\psi_0}D' I D \ket{\psi_0} = \int_{-\alpha}^{\alpha}\sum_{j\neq 0}d'_{j,0}d_{0,j}\delta(\omega-\omega_{j0})d\omega,
    \end{equation}
    and similarly, the restriction of the function to an interval $w$ as:
    \begin{equation}
        f_{w}(\omega) =  \int_{-\alpha}^{\alpha}\sum_{j\neq 0; \omega_{j0} \in w}d'_{j,0}d_{0,j}\delta(\omega-\omega_{j0})d\omega.
    \end{equation}
    
    Our goal is to determine the rough distribution of the coarse-grained function by measuring the frequency at which we observe bitstring $i$ in the ancilla register,
    \begin{equation}
        P(i) = \frac{1}{2|B|}\left(1+\frac{1}{\xi}{f_{i}}\right)
        \label{eq: prob dist of bins}
    \end{equation}
    with some suitably chosen $\xi = O(||f||)$ that is related to the subnormalization factor we found in the main text, so that $P(i)\leq 1$.

    We wish to determine inequalities amongst the $P(i)$ as this corresponds directly to the average value of the function over the interval $b_i$ and gives an ordering for how to proceed in the algorithm. The first step is to prepare a uniform superposition of Hadamard tests over different windows. Furthermore, we desire the property that the windows overlap in the $\delta$-region, where the ramping up and ramping down occur, so as to not have any gaps in our approximation.

    Since in the proof of lemma~\ref{lem: inequality test} the choices of $i$ and $j$ were arbitrary and the choice of $\tau = O(1/|B|)$ are sufficient to perform this inequality test for any $i,j \in [|B|]$. Therefore after sampling $N_s = \widetilde{O}(|B|^2)$ times we can determine inequalities amongst the bins with high confidence.
    
    There are two cases that can occur when determining these inequalities. The first (and simplest) case is that there is some $i' \in [|B|]$ such that for constant $\mu = 1/|B|,\tau$ $\text{Pr}(P(i') \geq (1+\tau)P(i))\leq \exp(-\tau^2/3|B|)$ for every other $i \in [|B|]$ with constant probability. For the other bins, we mark any that have a prominent response for later further resolution. The second possibility is that there is a set $S \subset [|B|]$ such that $\forall i,i' \in S$ $P(i)\sim P(i')$ and that $P(i) \gtrapprox P(j)$ whenever $i \in S$ and $j\notin S$ with constant probability. In fact, case two is a generalization of case one for the case when $|S| = 1$. In general, we will create subsets $S_i \subset [|B|]$ where for every $P_i,P_j \in S_i,$ $P_i \sim P_j$ in the sense that the inequality cannot be determined with constant probability. Additionally, we can order the $S_i$'s so that $P_i \in S_i$ and $P_j \in S_{i+1}$ $P_i < P_j$. We additionally have the property that $\cup_{i}S_{i} = [|B|]$.

\RestyleAlgo{ruled}
    \begin{algorithm}[hbt!]
        \caption{Algorithm for approximate sorting of bin heights}
        \KwData{windows $\{b_i\}_{i=0}^{|B|-1}$}
        \KwResult{$R$, array of (probabilistic) relations amongst bins}
        $\ket{P}\gets \sum_{i=0}^{|B|-1}\frac{1}{2|B|}\left(I + U_{b_i}\right)\ket{i}$\;
        $N_s \gets O(\log(|B|/\tau))$\;
        $S\gets$ zeros$(N_s)$\;
        \While{$i \leq N_s$}{
        $S \gets i\in[|B|]$ with $Pr(i) \sim \frac{1}{2|B|}\left(1+\frac{1}{||f||}f_{i}\right)$ \;
        $i += 1$\;
        }
        $P \gets$ zeros$(|B|)$\;
        \For{$j \in [N_s]$}{
            $P[S[j]] += \frac{1}{N_s}$\;
        }
        $R \gets $ zeros$(|B|,|B|)$\;
        \For{$i \in [|B|]$}{
        \While{$j < i$}{
        \eIf{$P[i] \gtrapprox (1+\tau)P[j]$}
        {
            $R[i,j] \gets 1$\;
            $R[j,i] \gets -1$;
        }{
            \If{$P[i] \approx P[j]$}
            {
                $R[i,j] \gets 0$\;
                $R[j,i] \gets 0$\;
            }
        }
        $j += 1$\;
        }
        $i += 1$\;
        }
        \label{alg: inequality testing}
    \end{algorithm}

    We detail the algorithm for case one, then show how to apply a similar procedure for the other cases. In this case, we assume that there is some prominent bin, $b_i'$ such that $P(i') \geq P(i)$ for all other $i$. In this case, we subdivide the interval $b_i'$ into $|B|$ bins, $\{ b_{i',0},b_{i',1},\ldots, b_{i',|B|-1}\}$ and form the uniform superposition over these bins of size $O(1/|B|^2)$ as per Alg. \ref{alg: inequality testing}. Then, we sample from 
    \begin{equation}
        P(i',j) = \sum_{j}\frac{1}{2|M|}\left(1+\frac{1}{\xi}f_{w_{i',j}}\right)
    \end{equation}
    for some suitable normalization $\xi$ so that $||f||/\xi \leq 1$. 
    
    At this stage it is possible that these bins satisfy the conditions of cases 1 or 2. If case 1, then there exists some prominent bin $j$ such that $P(i',j) > P(i',k)$ for every $(k\neq j) \in \{0,\ldots,|B|-1\}$, then repeat \ref{alg: inequality testing} as above (including the marking of less-prominent bins) until you have that $|B|^{-k} \leq \epsilon$ or until case 2 occurs. Then, remove this bin from future search iterations.

    \begin{figure}[h]
        \centering
        \includegraphics[width=.6\paperwidth]{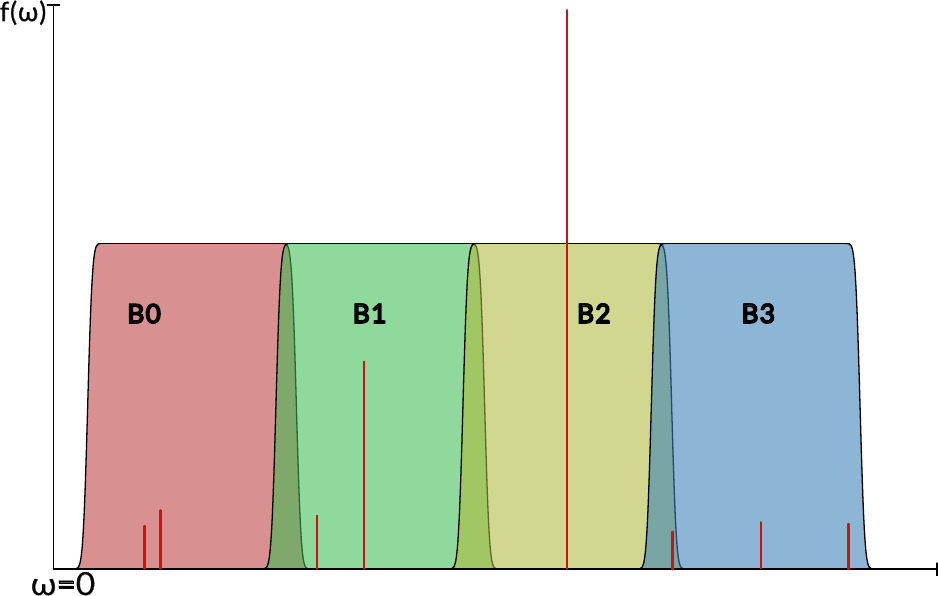}
        \caption{Figure showing the partially-overlapping bins covering the entire spectrum, as would occur in the first iteration of the algorithm.}
        \label{fig:bins1}
    \end{figure}

    If case 2 occurs, that is there exists some $S \subset [|B|]$ such that $\forall i,i' \in S P(i) \sim P(i')$ with some constant probability, then there are two possibilities. The first possibility occurs when there are at least two bins in $S$ that are neighboring, i.e. $i, i+1 \in S$ for some $i \in [|B|]$. In this case, it is possible that in the $\delta$-region where $b_i$ and $b_{i+1}$ overlap contains a large value of $f$. As an example, this occurs in the overlapping region between $B2$ and $B3$ in Fig \ref{fig:bins1}. To resolve this possibility, we combine $b_i$ and $b_{i+1}$ and form an equal superposition of $O(2|B|)$ bins spanning $b_{i} \cup b_{i+1}$. If there is a prominent bin in this case, apply case 1. 
    
    The other case occurs when elements of $S$ are disjoint. In this case we mark all the bins and choose a random $i$ and subdivide $b_i$ into $|B|$ bins. If there is a prominent value attained in one of the bins $b_{i,j}$ subdividing $b_i$ we take the bins $b_{i,j}$ and subdivide it further, otherwise re-apply. We can carry on this way until a bin of size $|B|^{-k} = O(\epsilon)$ is found, then remove this bin from future search iterations. 

    Since the function $f$ which describes our response is ``sparsely supported", we can determine the values in $A$ by repeating the above procedure until all the points in $A$ have been $\epsilon-$approximated. This procedure is efficient since by definition, points $x \notin A$ have $\epsilon-$small value, and so will only be observed with probability $O(\epsilon)$. Therefore, with high probability, our algorithm will find bins corresponding to points in $A$. Furthermore, since $|A| \in \widetilde{O}(1)$, this procedure need only be repeated $O(|A|)$ times. 

    We now prove that the number of queries to the block encoding of $H_0$ to determine some $x^* \in A$ saturates the Heisenberg limit. Prepare a uniform superposition of Hadamard tests as above, and choose the bin corresponding to the most frequently observed index (case 1), or one of the bins if there are multiple bins with similar values (case 2). This inequalities can be determined quickly, since they satisfy a Chernoff bound so the number of samples is logarithmic in the probability  We then subdivide this bin, once again choosing case 1 or case 2 respectively. We repeat this procedure until the bin size is of $O(\epsilon)$. 

    We analyze the complexity of performing this algorithm on a quantum computer. We assume that we have done $k-$iterations for each bin so that $|B|^{-k} < \epsilon$, our desired bin resolution. Since we started with the interval $[-\alpha, \alpha]$, starting from bins of size $O(1)$, we begin from bins of size $O(\alpha)$ and finish bins of size $O(1/\epsilon)$, requiring us to instead have $k = O(\log(\alpha/\epsilon))$. At each iteration $j$, bins of size $O(|B|^{-j})$ are formed, requiring that $\delta_j =  o(|B|^{-j})$. In turn, the approximating polynomial degree $d_j = \widetilde{O}(|B|^j)$, bounds the number queries to $U_H$ to implement the polynomial to the desired precision. Furthermore, each one of these $k$ runs must be sampled $O(|B|^2)$, for an additional $O(k|B|^2)$ queries. However, since $k = O(\log(\alpha/\epsilon))$ and $|B| = O(1)$, this contributes a non-dominant factor of $\widetilde{O}(1)$. Therefore the total query complexity to $U_H$ is
    \begin{equation}
        O(|B|^k) = \widetilde{O}(\alpha/\epsilon).
    \end{equation}

    However, since the algorithm relies on block encodings, the subnormalization factor scales the number of queries by the inverse success probability which is $O(||f||^2) = O(\alpha^2\beta^4)$, or through amplitude amplification the square root of the success probability $\alpha\beta^2$. In this case, we have
    \begin{equation}
        \widetilde{O}(\alpha^2\beta^2/\epsilon)
    \end{equation}
    queries to determine a point of $x \in A$ $\epsilon-$away in absolute value from the exact point.

    We repeat this process for each bin that we observer to have a probability greater than $\epsilon'$ an additional $\widetilde{O}(1)$ times to obtain an approximation $A'$ to the points in $A$. Then for each bin $w_i \in A'$, we measure $\bra{\psi_0}D' 1_{b_i}(H_0) D\ket{\psi_0}$ $O(1/\epsilon^2)$ times if using direct sampling, or an additional $O(1/\epsilon)$ times if using amplitude estimation, to obtain an $\epsilon-$approximation to $f_{i}$. 

    Then, since $f$ has the property of being sparsely supported, we know that 
    \begin{equation}
        f(x) = (1-\epsilon')\sum_{x \in A} f(x) + \epsilon' \sum_{x' \notin A} f(x).
    \end{equation}
    Then, defining our approximation $\widetilde{f}$
    \begin{equation}
        \widetilde{f}(x) = \sum_{w \in A'} \widetilde{f_w}(x),
    \end{equation}
    we can see that 
    \begin{equation}
        ||f(x) - \widetilde{f}(x)|| \leq (1-\epsilon')\epsilon|A| + \epsilon'
    \end{equation}
    therefore, if we choose $\epsilon = O(\epsilon'/|A|) = \widetilde{O}(\epsilon')$ then the total error is bounded by concentration parameter $\epsilon'$ as desired. The bounds for arbitrary $n$th order response can be found by replacing $||f|| = O\left(\alpha^{n}\beta^{n+1}\right)$ and $\epsilon \rightarrow \alpha^n/\epsilon^n$ to get the final complexity of $O\left(\alpha^{2n}\beta^{n+1}/\epsilon^n\right)$ queries to $U_H$.
\end{proof}

\subsubsection{Comparison of QPE vs Filtering method for computing response functions}
In the initial stages of this work, we first considered using phase estimation following application of the dipole operator and a \textsc{swap} test to accomplish a similar goal as this work.
To facilitate a fair comparison of the complexities of these two approaches, we will assume exact access to the ground state. As discussed above, this can be realized, for example, by starting from some high overlap estimate to the ground state, then using eigenstate filtering prepare an $O(1)$ overlap approximation to the ground state and compute a high-precision estimate of the ground energy $\lambda_0$. This will be the starting point of the algorithm.

The method of Hamiltonian simulation used depends on the input model. The asymptotically optimal method for Hamiltonian simulation is based on quantum signal processing, which relies on access to a block encoding of the problem Hamiltonian. The other most common algorithm is based on Trotterization, which directly implements the Hamiltonian as product of unitaries resulting from splitting terms in the Hamiltonian and simulating them separately. For the most straightforward comparison, we compare the asymptotic costs of quantum phase estimation with qubitization. The qubitization approach relies on an exponentially convergent approximation of the complex exponential known as the Jacobi-Anger expansion\cite{lowHamiltonianSimulationQubitization2019}. 

In that work, they show that using quantum signal processing in conjunction with the Jacobi-Anger expansion that the number of queries to the block encoding of $H$ to simulate a $d-$sparse Hamiltonian to time $t$ requires
\begin{equation}
    O\left(td \left\vert\left\vert H\right\vert\right\vert + \log(1/\gamma)/\log\log(1/\gamma)\right)
\end{equation}
queries to the block encoding of $H$, assuming $\left\vert\left\vert H \right\vert\right\vert = O(\alpha)$. To achieve $k$ bits of precision in the eigenvalues using phase estimation, the number of queries to the Hamiltonian simulation routine is $2^k$. Since in the second round of phase estimation the input quantum state generically has support over the entire spectrum of $H$. We desire that our approximation $\tilde{\lambda_i}$ is $\gamma$ accurate, i.e.
\begin{equation}
    \left\vert\lambda_i - \tilde{\lambda_i}\right\vert \leq \gamma,
\end{equation}
for each $\lambda_i$. However the eigenvalues have been rescaled by the subnormalization factor $\alpha$ from the block encoding, therefore we instead have
\begin{equation}
    \left\vert\lambda_i - \tilde{\lambda_i}\right\vert \leq \alpha \gamma'. 
\end{equation}
We desire that $\alpha \gamma' < \gamma$ so therefore $\gamma' < \gamma/\alpha$. Therefore we need $k$ digits of precision so that $2^{-k} < \gamma/\alpha \implies 2^k > \alpha/\gamma$. Therefore, we need $O(2^{k+1})$ queries to the Hamiltonian simulation routine. Taking the sparsity $d = O(1)$, the total cost, in terms of queries to the block encoding of $H$ is,
\begin{align*}
    &\sum_{i = 0}^{k-1}\left(2^i  \left\vert\left\vert H\right\vert\right\vert + \log(2^k)/\log\log(2^k)\right)\\
    &=\sum_{i = 0}^{k-1}\left(2^i  \left\vert\left\vert H \right\vert\right\vert + k/\log(k)\right)\\
    &=\left(2^{k}\left\vert\left\vert H \right\vert\right\vert + k^2/\log(k)\right)\\
    &= O\left(\frac{\alpha^2}{\gamma}+\log(\alpha/\gamma)^2/\log(k)\right)
\end{align*}
Then, we use a \textsc{swap} test on an additional register of $k$ ancilla qubits which prepares the uniform superposition of energies in the desired window. If we then perform a similar procedure to that given in the main text for searching through the spectrum, then we find that it takes an additional $O(\alpha/\gamma \epsilon)$ queries to find to determine the excitation energy to resolution $o(\gamma)$ and accuracy $\epsilon$.  If we use this algorithm in conjunction with amplitude estimation, we find that the total expected scaling becomes 
\begin{equation}
    \widetilde{O}(\alpha^2/\gamma^2\epsilon)
\end{equation}
So the filtering based algorithm improves over the phase estimation based algorithm by a factor of $1/\gamma$ in the number of queries to $U_H$.


\nocite{*}

\end{document}